\documentclass[11pt, reqno]{amsart}
\pdfoutput=1

\usepackage[top=4cm,bottom=4cm,left=3.5cm,right=3.5cm]{geometry}

\usepackage{amssymb,amsmath,amsthm}
\usepackage[UKenglish]{babel}
\usepackage{enumitem}
\setlist{itemsep=0pt, topsep=0.5em}
\setlist[description]{leftmargin=*, itemsep = 5pt}
\usepackage{tikz-cd}
\usepackage{mathtools}
\usepackage{graphicx}
\usepackage{mathrsfs}
\usepackage{stmaryrd}
\usepackage[scr=euler]{mathalfa}
\usepackage[hyphens]{url}
\usepackage{thmtools}
\usepackage{thm-restate}
\usepackage[shortcuts]{extdash}
\usepackage{microtype}

\providecommand{\noopsort}[1]{}

\theoremstyle{plain}
\newtheorem{theorem}{Theorem}[section]
\newtheorem*{theorem*}{Theorem}
\newtheorem{lemma}[theorem]{Lemma}
\newtheorem{proposition}[theorem]{Proposition}
\newtheorem{corollary}[theorem]{Corollary}

\newtheorem*{claim*}{Claim}
\newtheorem{fact}[theorem]{Fact}
\newtheorem*{fact*}{Fact} 

\theoremstyle{definition}
\newtheorem{definition}[theorem]{Definition}
\newtheorem{remark}[theorem]{Remark}

\newtheorem{example}[theorem]{Example}

\usepackage[breaklinks=true, hidelinks]{hyperref}

\newcommand{\Fraisse}{Fra\"{i}ss\'{e}}
\newcommand{\EF}{Ehrenfeucht\-/\Fraisse\ }
\renewcommand{\epsilon}{\varepsilon}
\renewcommand{\theta}{\vartheta}
\renewcommand{\phi}{\varphi}

\newcommand{\eg}{e.g.~}
\newcommand{\ie}{i.e.~}
\newcommand{\iec}{i.e.,~}
\newcommand{\Iec}{I.e.,~}

\newcommand{\down}{{\downarrow}\,} % Downward closure
\newcommand{\cvr}{\prec} % Covering relation
\DeclareMathOperator{\htf}{\mathrm{ht}} % Height of a point in a forest
\newcommand{\pit}{\pitchfork}
\newcommand{\br}[1]{\llbracket #1 \rrbracket} % pairs in back-and-forth systems
\newcommand{\brp}[1]{\br{#1}_{+}}

\renewcommand{\o}{\overline}
\DeclareMathOperator{\dom}{dom} % domain of an arrow

\DeclareMathOperator{\tp}{tp} %logical type of a tuple
\newcommand{\B}{\mathcal{B}} % Back-and-forth system
\newcommand{\sg}{\sigma}
\newcommand{\pref}{\sqsubseteq}
\newcommand{\N}{\mathbb{N}}
\newcommand{\tr}[1]{\llbracket #1 \rrbracket} % Standard translation
\renewcommand{\leq}{\leqslant}
\renewcommand{\geq}{\geqslant}
\newcommand{\Gf}{\mathfrak{G}} % Gaifman graph
\newcommand{\IMP}{\Rrightarrow}
\newcommand{\game}{\mathcal{G}}
\newcommand{\bisim}{\mathrel{\raisebox{1.5pt}{$\underline{\leftrightarrow}$}}} % bisimulation relation
\newcommand{\pbisim}{\mathrel{\raisebox{1.5pt}{$\underline{\rightarrow}$}}}

\newcommand{\Ek}{\mathbb{E}_k} % EF comonad
\newcommand{\EkI}{\Ek^{I}} % EF^{I} comonad
\newcommand{\Mk}{\mathbb{M}_k} % Modal comonad M_k
\newcommand{\Pk}{\mathbb{P}_k} %pebble comonad P_k
\newcommand{\PkI}{\Pk^{I}} %P_k^{I comonad}
\DeclareMathOperator{\Emb}{\mathbb{S}} % Strong subobjects
\DeclareMathOperator{\Path}{\mathbb{P}} % Path embeddings
\newcommand{\colim}{\operatornamewithlimits{colim}}
\newcommand{\FI}{\mathscr{F}^{I}}  % composite of J and F^{I}

\DeclareMathOperator{\C}{\mathscr{C}} % Generic (arboreal) category
\DeclareMathOperator{\A}{\mathscr{A}} % Generic category
\DeclareMathOperator{\T}{\mathscr{T}} % Category of trees
\newcommand{\CSstar}{\mathbf{Struct}_{\bullet}(\sg)} % Category of pointed \sg structures
\newcommand{\CSplus}{\mathbf{Struct}(\sg^I)} % Category of \sg^+ structures
\DeclareMathOperator{\D}{\mathscr{D}} % Full subcategory of the extensional category
\DeclareMathOperator{\Mod}{\mathbf{Mod}} % Category of models for a sentence
\newcommand{\CS}{\mathbf{Struct}(\sg)}
\newcommand{\EM}{\mathbf{EM}} % Eilenberg--Moore category

\newcommand{\LL}{\mathcal{L}} % Infinitary first-order logic
\newcommand{\ELL}{\exists\LL} % Existential infinitary first-order logic
\newcommand{\PLLk}{\prescript{+}{}{\LL}^{k} }

\newcommand{\EPLL}{\exists^{+}\LL} % Existential  positive infinitary first-order logic
\newcommand{\FO}{\mathrm{FO}} % FO logic
\newcommand{\EFO}{\exists\FO} % Existential FO logic
\newcommand{\PFO}{\prescript{+}{}{\FO}}
\newcommand{\EPFO}{\exists^+\FO} % Existential positive FO logic
\newcommand{\ML}{\mathrm{ML}} % Modal logic
\newcommand{\EML}{\exists\ML} % Existential modal logic
\newcommand{\PML}{\prescript{+}{}{\ML} }
\newcommand{\emb}{\rightarrowtail} % Embeddings (=strong monos)
\newcommand{\epi}{\twoheadrightarrow} % Epimorphism

\newcommand{\id}{\mathrm{id}} % Identity morphisms
\newcommand{\Q}{\mathscr{Q}} % Left class in the factorisation system
\newcommand{\M}{\mathscr{M}} % Right class in the factorisation system

\tikzcdset{
    relay arrow/.default = 10pt,
    relay arrow/.style = {
        rounded corners,
        to path = {
            \pgfextra{
                \def\sourcecoordinate{\pgfpointanchor{\tikztostart}{center}}
                \def\targetcoordinate{\pgfpointanchor{\tikztotarget}{center}}
                \pgfmathanglebetweenpoints{\sourcecoordinate}{\targetcoordinate}
                \edef\tempangle{\pgfmathresult}
                \pgftransformrotate{\tempangle}
                \pgfmathifthenelse{#1>0}{\tempangle+90}{\tempangle-90}
                \pgfcoordinate{tempcoord}{\pgfpointanchor{\tikztostart}{\pgfmathresult}}
            }
            (tempcoord)
            -- ([yshift=#1]tempcoord)
            -- ([yshift=#1]tempcoord-|\tikztotarget.center)\tikztonodes
            --(\tikztotarget)
        }
    }
}

\usetikzlibrary{calc}
\tikzset{hbend left/.style={
    to path={let \p1=($(\tikztotarget)-(\tikztostart)$),
        \p2=($(\tikztostart.north)-(\tikztostart.south)$),
        \p3=($(\tikztotarget.north)-(\tikztotarget.south)$),
        \n1={max(\y2,\y3)/2} in
    \ifdim\x1>0pt
     ([yshift=\n1]\tikztostart.east) to[bend left] ([yshift=\n1]\tikztotarget.west)
    \else
     ([yshift=-0.7*\n1]\tikztostart.west) to[bend left] ([yshift=-0.7*\n1]\tikztotarget.east)
    \fi
    }}}

\makeatletter
\renewcommand{\tocsection}[3]{%
  \indentlabel{\@ifnotempty{#2}{\bfseries\ignorespaces#1 #2\quad}}\bfseries#3}
\renewcommand{\tocsubsection}[3]{%
  \indentlabel{\@ifnotempty{#2}{\ignorespaces#1 #2\quad}}#3}
\newcommand\@dotsep{4.5}
\def\@tocline#1#2#3#4#5#6#7{\relax
  \ifnum #1>\c@tocdepth % then omit
  \else
    \par \addpenalty\@secpenalty\addvspace{#2}%
    \begingroup \hyphenpenalty\@M
    \@ifempty{#4}{%
      \@tempdima\csname r@tocindent\number#1\endcsname\relax
    }{%
      \@tempdima#4\relax
    }%
    \parindent\z@ \leftskip#3\relax \advance\leftskip\@tempdima\relax
    \rightskip\@pnumwidth plus1em \parfillskip-\@pnumwidth
    #5\leavevmode\hskip-\@tempdima{#6}\nobreak
    \leaders\hbox{$\m@th\mkern \@dotsep mu\hbox{.}\mkern \@dotsep mu$}\hfill
    \nobreak
    \hbox to\@pnumwidth{\@tocpagenum{\ifnum#1=1\fi#7}}\par
    \nobreak
    \endgroup
  \fi}
\AtBeginDocument{%
\expandafter\renewcommand\csname r@tocindent0\endcsname{0pt}
}
\def\l@subsection{\@tocline{2}{0pt}{2.5pc}{5pc}{}}
\makeatother

\begin{document}
\title[Existential and positive games: a comonadic and axiomatic view]{Existential and positive games: \\ a comonadic and axiomatic view}

\author{Samson Abramsky}
\address{Department of Computer Science, University College London, 66--72 Gower Street, London, WC1E 6EA, United Kingdom}\email{s.abramsky@ucl.ac.uk}

\author{Thomas Laure}
\address{D\'epartement d'Informatique, Ecole Normale Sup\'erieure, PSL University, 45 Rue d’Ulm, Paris 75005, France}
\email{thomas.laure@ens.psl.eu}

\author{Luca Reggio}
\address{Dipartimento di Matematica ``Federigo Enriques'', Universit\`a degli Studi di Milano, via Saldini 50, 20133 Milano, Italy}
\email{luca.reggio@unimi.it}

\thanks{Research supported by the EPSRC grant EP/V040944/1.}

\begin{abstract}
A number of model-comparison games central to (finite) model theory, such as pebble and \EF games, can be captured as comonads on categories of relational structures. 
In particular, the coalgebras for these comonads encode in a syntax-free way preservation of resource-indexed logic fragments, such as first-order logic with bounded quantifier rank or a finite number of variables.

In this paper, we extend this approach to existential and positive fragments (i.e., without universal quantifiers and without negations, respectively) of first-order and modal logic. 
We show, both concretely and at the axiomatic level of arboreal categories, that the preservation of existential fragments is characterised by the existence of so-called pathwise embeddings, while positive fragments are captured by a newly introduced notion of positive bisimulation.
As an application, we offer a new proof of an equi-resource Lyndon positivity theorem for (multi)modal logic.
\end{abstract}

\maketitle
\tableofcontents

%%%%%%%%%%%%%%%%%%%%%%%%%%%%%%%%%%%%%%%%%
%%%%%%%%%%%%%%%%%%%%%%%%%%%%%%%%%%%%%%%%%
\section{Introduction}
%%%%%%%%%%%%%%%%%%%%%%%%%%%%%%%%%%%%%%%%%
%%%%%%%%%%%%%%%%%%%%%%%%%%%%%%%%%%%%%%%%%

\subsection{Context}
There is a remarkable divide in computer science between semantics and structure on the one hand, exemplified by categorical semantics, and expressiveness and computational complexity on the other, exemplified by finite model theory. The present work is part of a broader programme aimed at bridging these two sides, called ``structure'' and ``power'' in \cite{ABRAMSKY2020}.

Finite model theory studies finite models of logical theories and is a notable example of the convergence of mathematical logic and computer science. It developed rapidly in the 1970s and 1980s due to its connections with database theory and complexity theory. A fundamental result embodying this connection is Fagin's theorem \cite{Fagin1974}. The latter initiated descriptive complexity, which aims to classify problems according to the language required to express them.
In this setting, a far-reaching idea is to stratify logic fragments by resource parameters such as quantifier rank or number of variables; cf.\ \eg\cite{Immerman1999}. This approach is reflected in the use of model-comparison games in which the two players, Spoiler and Duplicator, have restricted access to resources. Examples include \EF games with a finite number of rounds \cite{Ehrenfeucht1960,Fraisse1954} or pebble games \cite{Immerman1982}. Duplicator has a winning strategy if, and only if, two given structures cannot be distinguished in the corresponding resource-sensitive logic fragments.

From a categorical viewpoint, these games can be studied in a ``syntax-free'' way by means of game comonads, introduced in \cite{abramsky2017pebbling,DBLP:conf/csl/AbramskyS18}. The key idea is to consider plays in model-comparison games not as external artefacts, but instead as having a structure of their own. 
Generally speaking, each type of model-comparison game induces a comonad on the category of relational structures such that Duplicator winning strategies correspond to morphisms (or more generally spans of morphisms) in the category of coalgebras for the comonad. In the setting of game comonads, logical resources can be studied in a uniform way. This has led to the axiomatic notion of arboreal category~\cite{AR2022}, which captures the essential properties of the categories of coalgebras for game comonads, and the study of homomorphism preservation theorems therein~\cite{AR2024}. 

\subsection{Existential and positive games}
It is well known that if $A,B$ are finite relational structures such that $B$ satisfies all existential positive sentences satisfied by $A$ (\iec all sentences without $\forall$ and $\neg$), then there exists a homomorphism $A\to B$. In general, this does not hold when $A$ and $B$ are infinite, or when we restrict ourselves to a smaller logic fragment \eg by bounding the quantifier rank or the number of variables. However, as a consequence of the structural and resource-sensitive nature of game comonads, this characterisation can be recovered in this setting: the preservation of resource-sensitive existential positive fragments corresponds to the existence of a morphism between the co-free coalgebras on $A$ and $B$, respectively. This fact, which holds uniformly for finite and infinite structures, was observed at the concrete level in~\cite{abramsky2017pebbling,DBLP:conf/csl/AbramskyS18} and extended to the axiomatic level in~\cite{AR2022}. Similarly, the full (resource-sensitive) fragments are captured by the existence of a bisimulation between co-free coalgebras, in the sense of Definition~\ref{d:bisimulation}.

In this work we consider two intermediate logic fragments: the existential fragment, which forbids universal quantifiers and allows negations only at the atomic level, and the positive fragment, which forbids negations but allows universal quantifiers. We characterise the preservation of these fragments by first translating it into the existence of a winning strategy for Duplicator in an appropriate variant of the game, and then into the existence of appropriate morphisms in the category of coalgebras for the comonad.

Interestingly, while the existential and positive fragments behave in a somewhat symmetric (or ``dual'') way in terms of games, their coalgebraic characterisations are quite different. 
The existential fragments correspond to a forth-only variant of the game with a two-way (partial isomorphism) winning condition. The fact that it is a forth-only game corresponds to the absence of universal quantifiers, while the two-way winning condition captures the presence of negations on atoms. The existential fragments are captured by the notion of pathwise embedding introduced in~\cite{AS2021,AR2022}, which can be thought of as ``local embeddings'' between coalgebras.

In coalgebraic terms, positive fragments require a more subtle definition. At the level of games, they correspond to a back-and-forth variant of the game with a one-way (partial homomorphism) winning condition. We show that the preservation of positive fragments is captured by the existence of an asymmetric span in the category of coalgebras, which we call \emph{positive bisimulation}; see Definition~\ref{d:pos-bisim-concrete}.

For the benefit of the reader, we first explain our results in the concrete settings of the game comonads for \EF games, pebble games and bisimulation games. These results are then extended to the axiomatic level of arboreal categories.

\subsection{Outline of the paper}
The paper is structured as follows. In Sections~\ref{s:prel} and~\ref{s:comonads} we recall the necessary background on logic and games, and on game comonads, respectively. In Section~\ref{s:concrete-exs} we study preservation of existential resource-sensitive fragments from the comonadic viewpoint, and in Section~\ref{s:concrete-pos} we carry out an analogous analysis for positive fragments.
In the second part of the paper we adopt an axiomatic perspective. In Section~\ref{s:arboreal} we collect some basic facts about arboreal categories (some of which are new with this paper) and arboreal games, and in Sections~\ref{s:ex-arbor} and~\ref{s:pos-arbor} we provide axiomatic extensions of preservation results for existential and positive variants of arboreal games. Finally, in Section~\ref{s:preservation} we offer a new proof of an equi-resource positivity theorem for (multi)modal logic \`a la Lyndon, based on our previous results.

%%%%%%%%%%%%%%%%%%%%%%%%%%%%%%%%%%%%%%%%%
%%%%%%%%%%%%%%%%%%%%%%%%%%%%%%%%%%%%%%%%%
\section{Preliminaries}
\label{s:prel}
%%%%%%%%%%%%%%%%%%%%%%%%%%%%%%%%%%%%%%%%%
%%%%%%%%%%%%%%%%%%%%%%%%%%%%%%%%%%%%%%%%%
\subsection{Structures}\label{s:structures}

Consider a \emph{finite} relational vocabulary $\sigma$, \ie a first-order vocabulary containing finitely many relation symbols but no function symbols and no constant symbols. A \mbox{\emph{$\sigma$-structure}} is given by a set $A$ together with, for each relation symbol $R\in\sigma$ of arity $n\in\N$, an interpretation $R^A\subseteq A^n$. A \emph{homomorphism} of $\sigma$-structures $h\colon A \to B$ is a map that preserves the interpretations of the relation symbols, that is
\begin{equation}\label{eq:hom}
R^A(a_1,\dots,a_n) \ \Longrightarrow \ R^B(h(a_1),\dots,h(a_n))
\end{equation}
for each $R\in \sigma$ of arity $n$ and each tuple $(a_1,\ldots,a_n)\in A^n$.

We denote by 
\[
\CS
\]
the category of $\sigma$-structures and homomorphisms between them.
A morphism in the category $\CS$ is said to be an \emph{embedding} if it is injective and moreover it reflects the interpretations of the relation symbols, \ie the converse implication in eq.~\eqref{eq:hom} holds.

When working in the setting of modal logic, we will suppose that $\sigma$ is a \emph{(multi)modal vocabulary}, \ie it contains relation symbols of arity at most $2$. For each unary relation symbol $P\in\sigma$ we consider a corresponding propositional variable $p$, and for each binary relation symbol $R\in\sigma$ we consider corresponding modalities $\Diamond_R$ and $\Box_R$. We regard $\sigma$-structures as Kripke models for this (multi)modal logic, where the interpretations of unary relation symbols specify the worlds satisfying the corresponding propositional variables, and the interpretations of binary relations give the accessibility relations for the corresponding modalities. Modal formulas are evaluated over \emph{pointed} Kripke models, \ie pairs $(A,a)$ where $A$ is a Kripke model and $a\in A$. We let
\[
\CSstar
\]
be the category of pointed Kripke models and homomorphisms preserving the distinguished elements.

\subsection{Logic fragments}\label{s:logic-fragments}

We consider the infinitary logic $\LL_{\infty,\omega}(\sigma)$ obtained by extending usual first-order logic $\FO(\sigma)$ with conjunctions and disjunctions of arbitrary (set-sized) arity; however, each formula in $\LL_{\infty,\omega}(\sigma)$ can only use a finite number of (free or bound) variables. 
If no confusion arises, we omit reference to the vocabulary $\sigma$ and simply write, e.g., $\LL_{\infty,\omega}$ and $\FO$. 
Given a set of formulas $\mathbb{L}\subseteq\LL_{\infty,\omega}$ and $\sigma$-structures $A$ and $B$, write
\[
A\IMP_{\mathbb{L}} B
\]
if, for all sentences $\phi\in \mathbb{L}$, $A\models \phi$ implies $B\models \phi$. Further, we let $\equiv_{\mathbb{L}}$ be the symmetrization of the preorder $\IMP_{\mathbb{L}}$; \iec $A\equiv_{\mathbb{L}}B$ if and only if $A$ and $B$ satisfy the same sentences in~$\mathbb{L}$. If $(A,a)$ and $(B,b)$ are pointed structures, we take $(A,a)\IMP_{\mathbb{L}} (B,b)$ to mean that $A\models \phi(a)$ implies $B\models \phi(b)$ for all $\phi\in\mathbb{L}$ with at most one free variable. Following the usual practice, we write $\phi(x_{1},\ldots,x_{n})$ to denote that the free variables of $\phi$ are contained in $\{x_{1},\ldots,x_{n}\}$. If we want to emphasise the fact that $\phi$ is a sentence, we write $\phi(\varnothing)$.

We shall consider the following fragments of $\LL_{\infty,\omega}$:

\subsubsection*{Bounded quantifier rank}
Recall that the \emph{quantifier rank} of a formula is the maximum number of nested quantifiers appearing in it.
For each $k\in\N$, we denote by $\LL_k$ the fragment of $\LL_{\infty,\omega}$ consisting of formulas with quantifier rank at most $k$. Since $\sigma$ is finite, there are only finitely many formulas in $\LL_k$ up to logical equivalence, and so each formula in $\LL_k$ is equivalent to one that uses only finite conjunctions and disjunctions. Thus, $\LL_k$ can be identified with the fragment 
\[
\FO_{k}
\]
of first-order logic $\FO$ consisting of formulas with quantifier rank at most $k$.

\subsubsection*{Finite-variable logic}
For each positive integer $k$, we write 
\[
\LL^k
\] 
for the \emph{$k$-variable} fragment of $\LL_{\infty,\omega}$, consisting of the formulas using only (free or bound) variables $x_1,\ldots,x_k$. In general, $\LL^k$ contains infinitely many formulas up to logical equivalence, therefore it is a proper extension of $k$-variable first-order logic $\FO^k$. However, $\LL^k$ and $\FO^k$ have the same expressive power over \emph{finite} $\sigma$-structures, in the sense that for all finite $\sigma$-structures $A,B$ we have $A\equiv_{\LL^k} B$ if, and only if, $A\equiv_{\FO^k} B$.

\subsubsection*{Bounded modal depth}
Assume that $\sigma$ is a modal vocabulary. Whenever convenient, we shall tacitly identify the corresponding (multi)modal logic with a fragment of first-order logic via its \emph{standard translation} (see e.g.\ \cite[\S 2.4]{blackburn2002modal}), associating to each modal formula~$\phi$ a first-order formula $\tr{\phi}_x$ in a free variable~$x$. Then $(A,a)\models \phi$ according to Kripke semantics if, and only if, $A\models \tr{\phi}_x (a)$ in the usual model-theoretic sense.

Recall that the \emph{modal depth} of a modal formula is the maximum nesting of modalities in the formula. For each $k\in\N$, we let 
\[
\ML_k
\]
be the fragment of first-order logic $\FO$ consisting of (standard translations of) modal formulas with modal depth at most $k$. As with the fragments $\LL_k$, there are only finitely many formulas in $\ML_k$ up to logical equivalence.

\subsection{Model-comparison games}\label{s:games}

We recall the notions of model-comparison games corresponding to the logic fragments $\FO_k$, $\LL^k$ and $\ML_k$. These are two-player games, played by Spoiler and Duplicator on $\sigma$-structures $A$ and $B$.

\subsubsection*{\EF games}
In the $i$th round of the game, Spoiler chooses an element of one of the two structures, say $a_i\in A$, and Duplicator responds by choosing an element of the other structure, say $b_i\in B$. After $k$ rounds, Duplicator wins if the relation 
\begin{equation}\label{eq:rel-EF}
\{(a_i,b_i)\mid 1\leq i \leq k\}
\end{equation}
formed by the pairs of chosen elements is a \emph{partial isomorphism} between $A$ and $B$, \ie the assignment $a_i\mapsto b_i$ is an isomorphism between the induced substructures of $A$ and $B$, respectively, with carriers $\{a_1,\ldots,a_k\}$ and $\{b_1,\ldots,b_k\}$.

\begin{theorem}[\cite{Ehrenfeucht1960,Fraisse1954}]
    Duplicator has a winning strategy in the $k$-round \EF game between $A$ and $B$ if, and only if, $A\equiv_{\FO_k} B$.
\end{theorem}    

\subsubsection*{Pebble games}
Let $k$ be a positive integer. In the $k$-pebble game, each player has~$k$ pebbles available. In the $i$th round, Spoiler takes a pebble $p_i$ and places it on an element of one of the two structures, say $a_i\in A$ (if the pebble $p_i$ was already placed on another element, the effect is that of removing it and placing it on $a_i$). Duplicator responds by placing their corresponding pebble $p_i$ on an element of the other structure, say $b_i\in B$.
Duplicator wins the round if the relation determined by the current placings of the pebbles is a partial isomorphism between $A$ and $B$. Duplicator wins the $k$-pebble game if they have strategy that is winning after $n$ rounds, for all $n\in \N$. 

\begin{theorem}[\cite{Barwise1977,Immerman1982}]
    Duplicator has a winning strategy in the $k$-pebble game between~$A$ and~$B$ if, and only if, $A\equiv_{\LL^k} B$.
\end{theorem}        

\subsubsection*{Bisimulation games}
Suppose that $\sigma$ is a modal vocabulary. The bisimulation game for modal logic is played between pointed Kripke models $(A,a)$ and $(B,b)$.
The initial position of the game is $(a_0,b_0)\coloneqq(a,b)$. In the $i$th round, with $i>0$, if the current position is $(a_{i-1},b_{i-1})$, Spoiler chooses a binary relation symbol and one of the two structures, say $R\in\sigma$ and $A$, and an element $a_i\in A$ such that $R^A(a_{i-1},a_i)$. Duplicator responds with an element in the other structure, say $b_i\in B$, such that $R^B(b_{i-1},b_i)$. If Duplicator has no such response available, they lose. Otherwise, after $k$ rounds, Duplicator wins if for all unary relation symbols $P\in\sigma$ and all $0\leq i \leq k$ we have 
\begin{equation}\label{eq:winning-bisim}
P^A(a_i) \ \Longleftrightarrow \ P^B(b_i).
\end{equation}

\begin{theorem}[\cite{HM1980}]
    Duplicator has a winning strategy in the $k$-round bisimulation game between $(A,a)$ and $(B,b)$ if, and only if, $(A,a)\equiv_{\ML_k} (B,b)$.
\end{theorem}        

\section{Game comonads}
\label{s:comonads}

\subsection{Comonads}
We start by recalling the notion of a comonad:
\begin{definition}
A \emph{comonad (in Kleisli form)} on a category $\C$ is given by:
\begin{itemize}
    \item an object map $G\colon \mathrm{Ob}(\C)\to \mathrm{Ob}(\C)$,
    \item for each object $A$ of $\C$, an arrow $\epsilon_A \colon G A\to A$,
    \item a \emph{coextension} operation associating with each arrow $f\colon GA\rightarrow B$ an arrow $f^*\colon G A\rightarrow G B$,
\end{itemize}
such that the following equations are satisfied for all arrows ${f\colon GA\to B}$ and ${g\colon GB\to C}$:
\[
\epsilon_A^*=\id_{GA}, \ \ \epsilon_B\circ f^*=f, \ \ (g\circ f^*)^*=g^*\circ f^*.
\]
\end{definition}

Every comonad in Kleisli form $G$ on a category $\C$ induces an endofunctor $G\colon \C\to\C$ by setting $Gf \coloneqq (f\circ \epsilon_A)^*$ for each morphism $f\colon A\to B$ in $\C$. Furthermore, if we let $\delta_A\coloneqq \id_{GA}^*$ for each object $A$ of $\C$, then the tuple $(G,\epsilon,\delta)$ is a comonad in the traditional sense (see e.g.\ \cite[\S VI.1]{MacLane}), where the natural transformations $\epsilon\colon G\Rightarrow \id_{\C}$ and $\delta\colon G\Rightarrow G^2$ are, respectively, the counit and the comultiplication of the comonad.

\vspace{1em}
We will now define three comonads on $\sigma$-structures, corresponding to the three types of games introduced in Section~\ref{s:games}. Given finite sequences $s$ and $t$, we consider the \emph{prefix order} $\pref$ defined by 
\[
s\pref t
\]
if, and only if, there exists a (possibly empty) sequence $w$ such that $sw=t$, where $sw$ denotes the result of concatenating the sequences~$s$ and~$w$. 

\subsubsection*{The \EF comonad}
Fix $k\in\N$. For each $\sigma$-structure $A$, let $\Ek A$ be the set of all non-empty sequences of elements of $A$ of length $\leq k$, and consider the map 
\[
\epsilon_A\colon \Ek A \to A, \ \ \epsilon_A([a_1,\dots,a_i])\coloneqq a_i
\]
sending a sequence to its last element. For each relation symbol $R\in\sigma$ of arity $n$, its interpretation $R^{\Ek A}$ consists of the tuples of sequences $(s_1,\dots,s_n)\in (\Ek A)^n$ satisfying the following condition:

\begin{enumerate}[label=(E)]
\item\label{i:E} the sequences $s_1,\ldots,s_n$ are pairwise comparable with respect to $\pref$ (\iec for all $i,j\in\{1,\ldots,n\}$, either $s_i\pref s_j$ or $s_j\pref s_i$) and $R^A(\epsilon_A(s_1),\dots,\epsilon_A(s_n))$.
\end{enumerate} 
This gives an object map $\Ek$ from $\CS$ to itself, and condition~\ref{i:E} entails that the functions $\epsilon_A$ are homomorphisms. Finally, to define the coextension operation, for each homomorphism $f\colon \Ek A\rightarrow B$ we let 
\[
f^*([a_1,\dots,a_i])\coloneqq [b_1,\dots,b_i]
\]
where $b_{\ell}\coloneqq f([a_1,\dots,a_{\ell}])$ for all $\ell\in \{1,\ldots,i\}$.

These data define a comonad $\Ek$ on $\CS$, the \emph{\EF comonad}.

\subsubsection*{The pebbling comonad}
Fix a positive integer $k$. For each $\sigma$-structure $A$, we consider the set
$\Pk A \coloneqq (\{1,\dots,k\}\times A)^+$ of all finite non-empty sequences of pairs $(p_i,a_i)$, called \emph{moves}, where~$p_i$ is an integer between~$1$ and~$k$, and $a_i\in A$. We refer to $p_i$ as the \emph{pebble index} of the move. Similarly to the above, we define the maps
\[
\epsilon_A[(p_1,a_1),\dots,(p_i,a_i)] \coloneqq a_i
\]
sending a sequence of moves to the element of $A$ selected in the last move. For each relation symbol $R\in\sigma$ of arity $n$, we let $R^{\Pk A}$ be the set of all tuples $(s_1,\dots,s_n)\in (\Pk A)^n$ satisfying condition~\ref{i:E} above along with the following condition: 
\begin{enumerate}[label=(P)]
\item\label{i:P} for all $i,j\in\{1,\ldots,n\}$, if $s_i\pref s_j$ then the pebble index of the last move in~$s_i$ does not appear in the suffix of~$s_i$ in~$s_j$.
\end{enumerate} 
This defines an object map $\Pk$ on $\sigma$-structures, and condition~\ref{i:E} ensures that the $\epsilon_A$ are homomorphisms. The coextension operation is defined similarly to the above: for each homomorphism $f\colon \Pk A\rightarrow B$, we set
\[
f^*([(p_1,a_1),\dots,(p_i,a_i)])\coloneqq [(p_1,b_1),(p_2,b_2),\dots,(p_i,b_i) ]
\]
where $b_{\ell}\coloneqq f([(p_1,a_1),\dots,(p_{\ell},a_{\ell})])$ for all $\ell\in \{1,\ldots,i\}$.

This yields a comonad $\Pk$ on $\CS$, the \emph{pebbling comonad}.

\subsubsection*{The modal comonad}
Let $\sigma$ be a modal vocabulary and fix $k\in\N$. Given a pointed Kripke model $(A,a)$, we denote by $\Mk(A,a)$ the set of all paths in $A$ of length at most~$k$ starting from $a$:
\begin{equation}\label{eq:path-Mk}
a=a_0 \xrightarrow{R_1} a_1 \xrightarrow{R_2} a_2 \to\cdots\xrightarrow{R_i}a_i.
\end{equation}
(That is, $R^A_{j}(a_{j-1},a_j)$ for all $j\in\{1,\ldots, i\}$.)

For each unary relation symbol $P\in\sigma$ its interpretation in $P^{\Mk(A,a)}$ is the set of all paths whose last elements are in $P^A$, and for each binary relation symbol $R\in\sigma$ its interpretation $R^{\Mk(A,a)}$ consists of the pairs of paths $(s,t)$ such that $t$ is obtained by extending $s$ by one step along $R$. 
The distinguished element of $\Mk(A,a)$ is the trivial path $(a)$ of length $0$.
This yields an object map $\Mk$ on pointed Kripke models, which sends a model to its \emph{$k$-unravelling} (cf.\ \cite[Definition~4.51]{blackburn2002modal}).
The maps
\[
\epsilon_{(A,a)}\colon \Mk(A,a)\to (A,a)
\]
sending a path to its last elements are homomorphisms. Finally, the coextension operation associates with a homomorphism $f\colon \Mk(A,a)\to (B,b)$ the homomorphism
\[
f^*\colon \Mk(A,a)\to\Mk(B,b)
\]
that sends a path as in eq.~\eqref{eq:path-Mk} to the path
\[
b=b_0 \xrightarrow{R_1} b_1 \xrightarrow{R_2} b_2 \to\cdots\xrightarrow{R_i}b_i
\]
where $b_{\ell}\coloneqq f(a_0 \xrightarrow{R_1}\cdots\xrightarrow{R_{\ell}}a_{\ell})$ for all $\ell\in\{0,\ldots,i\}$.

This gives a comonad $\Mk$ on $\CSstar$, the \emph{modal comonad}.

\subsection{Coalgebras}\label{s:coalgebras}
Every comonad $G$ on a category $\C$ induces a notion of coalgebras (sometimes referred to as $G$-coalgebras), which we now recall.

\begin{definition}
A \emph{coalgebra} for $G$ is a pair $(X,\alpha)$ consisting of an object~$X$ of $\C$ and a morphism $\alpha \colon X \rightarrow G X$ in $\C$ such that the following diagrams commute: 

\begin{center}
\begin{tikzcd}
    X \arrow{r}{\alpha} \arrow{d}[swap]{\alpha} & G X\arrow{d}{\delta_X \, = \, \id_{GX}^*} \\
    G X\arrow{r}{G \alpha} & G^2 X
\end{tikzcd}
\ \ \ \ \ 
\begin{tikzcd}
    X \arrow{r}{\alpha}\arrow{rd}[swap]{\id_X} & G X \arrow{d}{\epsilon_X}\\ 
    & X
\end{tikzcd}
\end{center}
A \emph{coalgebra morphism} $(X,\alpha)\to(Y,\beta)$ is a morphism $h\colon X\rightarrow Y$ in $\C$ making the following square commute:
\begin{center}
    \begin{tikzcd}
        X\arrow{r}{\alpha} \arrow{d}[swap]{h} & G X \arrow{d}{G h} \\
        Y \arrow{r}{\beta} & G Y
    \end{tikzcd}
\end{center}
    Finally, the \emph{Eilenberg--Moore category} $\EM(G)$ of the comonad $G$ is the category whose objects are $G$-coalgebras and whose arrows are coalgebra morphisms.
\end{definition}

The forgetful functor $L\colon \EM(G)\to \C$ that sends a coalgebra $(X,\alpha)$ to $X$ has a right adjoint 
\[
F\colon \C\to \EM(G), 
\]
which sends an object $A$ of $\C$ to the co-free coalgebra $(GA, \delta_A)$. We recall from~\cite{DBLP:conf/csl/AbramskyS18} the concrete description of the adjunction $L\dashv F$, typically referred to as \emph{Eilenberg--Moore adjunction}, in the case of the comonads $\Ek$, $\Pk$ and $\Mk$. To this end, we make use of the following terminology. 

\begin{definition}
The \emph{Gaifman graph} $\Gf(A)$ of a $\sigma$-structure $A$ is the (undirected) graph defined as follows:
\begin{itemize}
\item the vertices of $\Gf(A)$ are the elements of $A$,
\item elements $a,a'\in A$ are adjacent in $\Gf(A)$, denoted by $a\frown a'$, if they are distinct and both appear in some tuple of related elements $\vec{a}\in R^A$, for some $R\in\sigma$.  
\end{itemize}
\end{definition}

\begin{definition}\label{d:forests}
Let $(X,\leq)$ be a poset. 
\begin{enumerate}[label=(\arabic*)]
\item $(X,\leq)$ is a \emph{forest} if, for each $x\in X$, the set 
\[\down x \coloneqq \{y\in X\mid y\leq x\}\] is finite and linearly ordered by $\leq$. A minimal element of a forest is called a \emph{root}; a forest with exactly one root is a \emph{tree}. 
\item Given a forest $(X,\leq)$, the \emph{height} of $x\in X$ is 
$
\htf(x)\coloneqq |\down x|-1
$
(in particular, roots have height $0$). The \emph{height} of the forest $(X,\leq)$ is $\sup{\{\htf(x)\mid x\in X\}}$.
\item The \emph{covering relation} $\cvr$ associated with the partial order $\leq$ is defined by $x\cvr y$ if, and only if, $x<y$ and there is no $z$ such that $x<z<y$. 
\item A \emph{forest morphism} between forests $(X,\leq)$ and $(X',\leq')$ is a map $f\colon X\to X'$ that sends roots to roots and preserves the covering relation (\ie $x\cvr y$ implies $f(x)\cvr' f(y)$ for all $x,y\in X$).\footnote{Equivalently, it is a monotone map that preserves the height of elements.}
\end{enumerate}
\end{definition}

\subsubsection*{$\Ek$-coalgebras}
The Eilenberg--Moore category $\EM(\Ek)$ for the \EF comonad $\Ek$ is isomorphic to the following category:
\begin{itemize}
\item objects: forest-ordered structures of height $\leq k$, \ie pairs $(A,\leq)$ where $A$ is a $\sigma$-structure and $\leq$ is a forest order on $A$ of height $\leq k$, satisfying
\begin{enumerate}[label=($\mathsf{E}$)]
\item\label{i:E-bis} for all $a,a'\in A$, if $a\frown a'$ in $\Gf(A)$ then $a$ and $a'$ are comparable with respect to $\leq$ (\ie $a$ and $a'$ belong to the same branch of the forest).
\end{enumerate} 
\item morphisms: homomorphisms of $\sigma$-structures that are also forest morphisms.
\end{itemize}

The right adjoint to the forgetful functor $L_{k}\colon \EM(\Ek)\to\CS$, which forgets the forest order, is the functor $F_{k}\colon \CS\to \EM(\Ek)$ that sends a structure $A$ to the forest ordered structure $(\Ek A,\pref)$, where $\pref$ is the prefix order.

\subsubsection*{$\Pk$-coalgebras}
The Eilenberg--Moore category $\EM(\Pk)$ for the pebbling comonad $\Pk$ can be identified, up to isomorphism, with the following category:

\begin{itemize}
\item objects: $k$-pebble forest-ordered structures, \ie tuples $(A,\leq,p)$ where $A$ is a $\sigma$-structure,  $\leq$ is a forest order on $A$ and $p\colon A\to \{1,\dots,k\}$ is a function (called \emph{pebbling function}) satisfying condition~\ref{i:E-bis} above along with
\begin{enumerate}[label=($\mathsf{P}$)]
\item\label{i:P-bis} for all $a,a'\in A$, if $a\frown a'$ in $\Gf(A)$ and $a<a'$ in the forest order, then $p(a)\neq p(x)$ for all $x$ such that $a<x\leq a'$.
\end{enumerate} 
\item morphisms: homomorphisms of $\sigma$-structures that are also forest morphisms and preserve the pebbling functions.
\end{itemize}

The forgetful functor $L_{k}\colon \EM(\Pk)\to\CS$ forgets both the forest order and the pebbling function, and its right adjoint $F_{k}\colon \CS\to\EM(\Pk)$ sends a structure $A$ to $(\Pk A,\pref,p_A)$, where the pebbling function is given by $p_A([(p_1,a_1), \dots , (p_n,a_n)]) \coloneqq p_n$.

\subsubsection*{$\Mk$-coalgebras}
The Eilenberg--Moore category $\EM(\Mk)$ for the modal comonad $\Mk$ is isomorphic to the following category:

\begin{itemize}
\item objects: \emph{synchronization trees} of height $\leq k$, \ie tuples $(A,a,\leq)$ such that $(A,a)$ is a pointed Kripke model, $\leq$ is a tree order on $A$ of height $\leq k$ with root $a$, and 
\begin{enumerate}[label=($\mathsf{M}$)]
\item\label{i:M-bis} for all $x,y\in A$, $x\cvr y$ if and only if $R^A(x,y)$ for some unique $R\in\sigma$.
\end{enumerate} 
\item morphisms: homomorphisms of Kripke models. 
\end{itemize}

Note that the tree order of a synchronization tree is ``definable'' and hence preserved by any homomorphism of Kripke models. Thus the forgetful functor $L_{k}\colon \EM(\Mk)\to\CSstar$ is fully faithful and can be identified with the inclusion into $\CSstar$ of the full subcategory defined by synchronization trees of height $\leq k$. Its right adjoint $F_{k}\colon \CSstar\to \EM(\Mk)$ sends a Kripke model to its $k$-unravelling.

\begin{remark}
In fact, the modal comonads $\Mk$ are \emph{idempotent}, meaning that their comultiplications $\delta$ are natural isomorphisms. In turn, idempotent comonads on a category~$\C$ correspond precisely to \emph{coreflective subcategories} of~$\C$, \ie full subcategories~$\D$ such that the inclusion functor $\D\hookrightarrow \C$ has a right adjoint.
\end{remark}

\subsection{Logical equivalences}

In \cite{abramsky2017pebbling,DBLP:conf/csl/AbramskyS18} it was shown how game comonads can be used to capture, in a categorical fashion, preservation of various fragments and extensions of the logics $\FO_{k}$, $\LL^{k}$ and~$\ML_{k}$. We recall here the case of the existential positive fragments and the full fragments, and outline the way in which the equality symbol is handled in the framework of game comonads.

\subsubsection*{Existential positive fragments} 
Given a collection of formulas $\mathbb{L}\subseteq \LL_{\infty,\omega}$, consider its \emph{existential positive} fragment
\[
\exists^{+}\mathbb{L}
\] 
consisting of formulas with no universal quantifiers and no negations; for modal logic, this corresponds to barring the use of modalities~$\Box_{R}$ and negations. The next result shows that preservation of existential positive fragments is captured by the homomorphism preorder in the Eilenberg--Moore category for the corresponding comonad (for a proof, cf.\ \cite{abramsky2017pebbling,DBLP:conf/csl/AbramskyS18} and also \cite{AS2021}).\footnote{We could equivalently use the homomorphism preorder in the \emph{Kleisli category} of the comonad, which can be identified with a full subcategory of the Eilenberg--Moore category.} Recall that, for a comonad~$G_{k}$ on a category $\C$, we denote by $F_{k}\colon \C\to \EM(G_{k})$ the right adjoint to the forgetful functor $L_{k}\colon \EM(G_{k})\to\C$.

\begin{theorem}\label{th:ep-Kleisli-arrows}
The following are equivalent for all (pointed) structures $A$ and $B$:
\begin{enumerate}[label=(\roman*)]
\item $A\IMP_{\EPFO_{k}} B$ if, and only if, there exists an arrow $F_{k} A\to F_{k} B$ in $\EM(\Ek)$. 
\item $A\IMP_{\EPLL^{k}} B$ if, and only if, there exists an arrow $F_{k} A\to F_{k} B$ in $\EM(\Pk)$. 
\item $(A,a)\IMP_{\exists^{+}\ML_{k}} (B,b)$ if, and only if, there exists an arrow $F_{k} (A,a)\to F_{k} (B,b)$ in $\EM(\Mk)$. 
\end{enumerate}
\end{theorem}

\subsubsection*{The equality symbol}
In their basic form, game comonads capture fragments of logics without equality. This is sufficient for modal logic, since the image of the standard translation is contained in the equality-free fragment of $\FO$, and was also sufficient for Theorem~\ref{th:ep-Kleisli-arrows}, since the fragments $\EPFO_{k}$ and $\EPLL^{k}$ admit equality elimination. To model logics \emph{with} equality, such as the fragments $\FO_{k}$ and $\LL^{k}$, we proceed as follows.

Given a relational vocabulary $\sigma$, consider a fresh binary relation symbol $I$ and the expanded vocabulary
\[
\sigma^{I}\coloneqq \sigma \cup \{I\}.
\]
There is a fully faithful functor $J\colon \CS \to \CSplus$ that sends a $\sg$-structure to the $\sg^{I}$-structure obtained by interpreting $I$ as the diagonal relation. This functor has a left adjoint $H\colon \CSplus\to\CS$ that sends a $\sg^{I}$-structure $A$ to the quotient of (the $\sg$-reduct of) $A$ by the equivalence relation generated by $I^{A}$. Thus, the functor~$J$ introduces (the interpretation of) the equality symbol, while $H$ eliminates it. Since a generic game comonad $G$ on $\CS$ is defined uniformly for any relational vocabulary $\sg$, we can consider its variant $G^{I}$ over $\sg^{I}$-structures. If $L^{I}\dashv F^{I}$ is the Eilenberg--Moore adjunction associated to $G^{I}$, we have the following situation:
\[\begin{tikzcd}[column sep=0.8em]
\CS \arrow[hbend left]{rr}{J} & \text{\footnotesize{$\top$}} & \CSplus \arrow[hbend left]{rr}{F^{I}} \arrow[hbend left]{ll}{H} & \text{\footnotesize{$\top$}} & \EM(G^{I}) \arrow[hbend left]{ll}{L^{I}}
\end{tikzcd}\]
Let $\FI \coloneqq F^{I}J\colon \CS\to \EM(G^{I})$. The composite adjunction
\[\begin{tikzcd}[column sep=0.8em]
\CS \arrow[hbend left]{rr}{\FI} & \text{\footnotesize{$\top$}} & \EM(G^{I}) \arrow[hbend left]{ll}{HL^{I}}
\end{tikzcd}\]
is not comonadic, but is the one that captures logics with equality.

Consider e.g.\ the \EF comonad $\Ek$. Note that, for any $\sg$-structure~$A$, the interpretation of the relation $I$ in the $\sg^{I}$-structure $\Ek J(A)$ consists of those pairs $(s,t)$ of sequences such that:
\begin{enumerate}[label=(\roman*)]
\item $s,t$ are comparable in the prefix order $\sqsubseteq$, and
\item their last elements are equal, \ie $\epsilon_{A}(s)=\epsilon_{A}(t)$.
\end{enumerate}
In other words, the relation $I^{\Ek J(A)}$ detects repeating sequences. Given $\sg$-structures~$A$ and~$B$, morphisms 
\[
\FI(A) \to \FI(B)
\]
in $\EM(\EkI)$ are in bijection with homomorphisms $\EkI J(A)\to J(B)$ in $\CSplus$, which in turn correspond precisely to homomorphisms $f\colon \Ek (A)\to B$ in $\CS$ such that
\[
\text{$\forall s,t\in \Ek(A)$, if $s\sqsubseteq t$ and $\epsilon_{A}(s)=\epsilon_{A}(t)$ then $f(s)=f(t)$},
\]
called \emph{$I$-morphisms} (see~\cite{abramsky2017pebbling} or \cite[\S 4]{AS2021}). In the case of the pebbling comonad, $I$-morphisms are defined in a similar fashion.

\subsubsection*{The full fragments}

To capture equivalence in the full fragments, we introduce the following notions, which rely on the concrete description of the Eilenberg--Moore categories of the comonads $\Ek$, $\Pk$ and $\Mk$ provided in Section~\ref{s:coalgebras}.

For the remainder of this section, $\A$ will denote one of the categories $\EM(\Ek)$, $\EM(\Pk)$ or $\EM(\Mk)$.

\begin{definition}\label{d:pathwise} \phantomsection
\mbox{}
        \begin{enumerate}[label=(\arabic*)]
    \item\label{i:path} A \emph{path} in $\A$ is an object whose forest order is finite and linear, \ie it consists of a single branch. Paths are denoted by $P,Q$ and variations thereof.
    \item\label{i:embedding} An \emph{embedding} in $\A$, denoted by $\emb$, is a morphism whose underlying homomorphism of $\sg$-structures is an embedding (in the sense of Section~\ref{s:structures}).
    
    \item\label{i:path-embedding} A \emph{path embedding} $P\emb X$ in $\A$ is an embedding whose domain is a path. 
    
    \item\label{i:pathwise-embedding} A morphism $f\colon X\to Y$ in $\A$ is a \emph{pathwise embedding} if, for every path embedding $m\colon P\emb X$, the composite $f\circ m \colon P \to Y$ is a path embedding.
    \end{enumerate}
\end{definition}

Note that pathwise embeddings may fail to be injective. Thus, in general, the class of embeddings is properly contained in the class of pathwise embeddings. 

We next introduce open morphisms in $\A$ in terms of an appropriate \emph{path lifting property} which captures the back-and-forth nature of games for full fragments.

\begin{definition}
    \label{d:open}
  A morphism $f \colon X\to Y$ in $\A$ is said to be \emph{open} if every commutative square as displayed below,
      \begin{center}
        \begin{tikzcd}
P \arrow[rightarrowtail]{r}{m} \arrow[rightarrowtail]{d}[swap]{i} & X \arrow{d}{f} \\
Q \arrow[rightarrowtail]{r}{n} \arrow[dashed,rightarrowtail]{ur}[description]{d} & Y
\end{tikzcd}
    \end{center}
where $i,m,n$ are path embeddings, admits a \emph{diagonal filler}, \ie an arrow $d\colon Q\to X$ making the two triangles commute (note that such a $d$, if it exists, is an embedding).
\end{definition}

\begin{example}\label{ex:p-morphisms-open-maps}
Recall that a function $f\colon X\to Y$ between posets is a \emph{p-morphism} if it is monotone and, for all $x\in X$ and $y\in Y$,
\[
f(x) \leq y \ \Longrightarrow \ \exists x'\in X \text{ such that } x\leq x' \text{ and } f(x') = y.
\] 
This is a special case of the notion of p-morphism (also called \emph{bounded} morphism) between Kripke frames, which is central to modal logic; see e.g.\ \cite[p.~140]{blackburn2002modal}.

When $X$ and $Y$ are forests, a forest morphism $h\colon X\to Y$ is a p-morphism if, and only if, for all $x\in X$ and $y\in Y$ it satisfies
\[
f(x) \cvr y \ \Longrightarrow \ \exists x'\in X \text{ such that } x\cvr x' \text{ and } f(x') = y.
\] 
In the same spirit of Definition~\ref{d:open}, we can define a morphism $f\colon X\to Y$ in the category of forests and forest morphisms (cf.~Definition~\ref{d:forests}) to be \emph{open} if every commutative square
      \begin{center}
        \begin{tikzcd}
P \arrow{r} \arrow{d} & X \arrow{d}{f} \\
Q \arrow{r}  & Y
\end{tikzcd}
    \end{center}
with $P$ and $Q$ finite chains admits a diagonal filler (just note that every forest morphism whose domain is a chain is injective). It turns out that a forest morphism is open if and only if, it is a p-morphism. Furthermore, a forest morphism between non-empty trees is open if and only if it is a surjective p-morphism.
\end{example}

\begin{definition}\label{d:bisimulation}
A \emph{bisimulation} between objects $X,Y$ of $\A$ is a span of open pathwise embeddings connecting $X$ and $Y$:
\[
X \leftarrow Z \rightarrow Y.
\]
We say that $X$ and $Y$ are \emph{bisimilar}, and write $X\bisim Y$, if there exists a bisimulation between them.
\end{definition}

The next result, cf.~\cite{AR2021icalp} for a proof, states that the bisimilarity relation captures precisely equivalence in the full fragments:
\begin{theorem}\label{th:full-bisimulations}
The following statements hold for all (pointed) structures~$A$ and~$B$:
\begin{enumerate}[label=(\arabic*)]
\item $A\equiv_{\FO_{k}} B$ if, and only if, $\FI_{k} A\bisim \FI_{k} B$ in $\EM(\EkI)$. 
\item $A\equiv_{\LL^{k}} B$ if, and only if, $\FI_{k} A\bisim \FI_{k} B$ in $\EM(\PkI)$. 
\item $(A,a)\equiv_{\ML_{k}} (B,b)$ if, and only if, $F_{k} (A,a)\bisim F_{k} (B,b)$ in $\EM(\Mk)$. 
\end{enumerate}
\end{theorem}

\begin{remark}\label{rem:ope-bisim-modal}
Open pathwise embeddings in $\EM(\Mk)$ are (surjective) p-morphisms in the sense of modal logic. Therefore, the last item of Theorem~\ref{th:full-bisimulations} corresponds to the well-known fact that two Kripke models are $\ML_{k}$-equivalent precisely when their $k$-unravellings are bisimilar (in the sense of modal logic).
\end{remark}

%%%%%%%%%%%%%%%%%%%%%%%%%%%%%%%%%%%%%%%%%
%%%%%%%%%%%%%%%%%%%%%%%%%%%%%%%%%%%%%%%%%
\section{Existential fragments and pathwise embeddings}
\label{s:concrete-exs}
%%%%%%%%%%%%%%%%%%%%%%%%%%%%%%%%%%%%%%%%%
%%%%%%%%%%%%%%%%%%%%%%%%%%%%%%%%%%%%%%%%%

Given a logic fragment $\mathbb{L}\subseteq \LL_{\infty,\omega}$, we consider the \emph{existential} fragment
\[
\exists \mathbb{L}
\]  
of $\mathbb{L}$, consisting of formulas with no universal quantifiers (or no modalities $\Box_R$, in the case of modal logic), and negations only applied to atomic formulas. 

Our first main result consists in a comonadic characterisation of preservation of the existential fragments $\EFO_{k}$, $\ELL^{k}$ and $\EML_{k}$ in terms of the existence of a pathwise embedding between the relevant coalgebras. More precisely, we show that:

\begin{theorem}
The following statements hold for all (pointed) structures $A$ and $B$:
\begin{enumerate}[label=(\arabic*)]
\item $A\IMP_{\EFO_{k}} B$ if, and only if, there exists a pathwise embedding $\FI_{k}(A)\to \FI_{k}(B)$ in $\EM(\EkI)$. 
\item $A\IMP_{\ELL^{k}} B$ if, and only if, there exists a pathwise embedding $\FI_{k}(A)\to \FI_{k}(B)$ in $\EM(\PkI)$. 
\item $(A,a)\IMP_{\EML_{k}} (B,b)$ if, and only if, there exists a pathwise embedding $F_{k}(A,a)\to F_{k}(B,b)$ in $\EM(\Mk)$. 
\end{enumerate}
\end{theorem}
We establish this in two steps:
\begin{enumerate}[label=(\roman*)]
\item\label{i:first-step} first, we show that the existential fragments are captured by \emph{existential} (forth-only) variants of the corresponding games;
\item\label{i:second-step} then we prove that the existence of a Duplicator winning strategy in the existential game is equivalent to the existence of a pathwise embedding between the corresponding co-free coalgebras.
\end{enumerate}

For each of the three logic fragments, item~\ref{i:first-step} relies on a routine argument, so we will only sketch the main ideas and provide detailed proofs in the appendix. Item~\ref{i:second-step} is a special case of a more general result which holds at an axiomatic level, see Section~\ref{s:ex-arbor}, so in this section we will only provide the reader with the concrete intuitions behind the axiomatic arguments, deferring the full proofs to that section.

\subsection{Existential \EF game}
\label{ss:ex-EF}
Let us consider the fragment $\EFO_k$ consisting of existential first-order formulas with quantifier rank at most $k$.
The intuition is that preserving negated atoms is, by contraposition, equivalent to reflecting relations. Thus we define the \emph{existential \EF game} from $A$ to $B$ to be the variant of the \EF game in which Spoiler plays only in $A$ and Duplicator in $B$. The winning condition for Duplicator is the same, namely that the resulting relation, see eq.~\eqref{eq:rel-EF}, is a partial isomorphism between~$A$ and~$B$.

This definition is justified by the following result:

\begin{restatable}{proposition}{propexEFlogic}
\label{p:ex-EF-logic}%
The following statements are equivalent for all structures $A, B$:
\begin{enumerate}[label=(\roman*)]
\item\label{i:exist-pres-EF} $A\IMP_{\EFO_{k}} B$, \ie every sentence of $\EFO_{k}$ that is true in $A$ is true in $B$.
\item\label{i:win-exis-EF} Duplicator has a winning strategy in the $k$-round existential \EF game from $A$ to $B$.
\end{enumerate}
\end{restatable}

\begin{proof}
This is a routine argument. We only sketch a proof, for more details see the appendix.

\ref{i:exist-pres-EF} $\Rightarrow$ \ref{i:win-exis-EF}
For each $\sigma$-structure $C$ and tuple of elements $(c_1,\dots,c_i)$ in $C$ such that $0\leq i \leq k$, consider its \emph{$\EFO_{k}$-type}
\[
\tp_C(c_1,\dots,c_i)\coloneqq \{\phi(x_1,\dots,x_i) \in \EFO_{k-i} \mid C\models \phi (c_1,\dots,c_i)\}.
\] 
One can prove that, for all $i\geq 1$ and all tuples of elements $a_1,\dots,a_{i-1}$ and $b_1,\dots,b_{i-1}$ in $A$ and $B$, respectively, if $\tp_A(a_1,\dots,a_{i-1})\subseteq \tp_B(b_1,\dots,b_{i-1})$ then for each $a_i\in A$ there exists $b_i\in B$ such that $\tp_A(a_1,\dots,a_i)\subseteq \tp_B(b_1,\dots,b_i)$.
In turn, this provides a winning strategy for Duplicator in the $k$-round existential \EF game from $A$ to $B$.

\ref{i:win-exis-EF} $\Rightarrow$ \ref{i:exist-pres-EF}
Let $\phi(\varnothing) \in \exists \FO_k$. Since a formula of quantifier rank $k$ can fruitfully use at most $k$ variables, we can $\alpha$-rename the (bound) variables of $\phi$ to $x_1,\dots,x_k$, and we can do so in such a way that the quantifiers $\exists x_i$ appear in $i$-decreasing order within a quantifier branch (where we identify a formula with its syntax tree). 

One can then show by structural induction that for each subformula $\psi(x_1,\dots,x_i)$ of~$\phi$, if $A\models \psi(a_1,\dots,a_i)$ then $B\models \psi(b_1,\dots,b_i)$ where $[b_1,\dots,b_i]$ is Duplicator's answer to Spoiler's play $[a_1,\dots,a_i]$.
For $\psi=\phi$, we get that $A\models \phi$ implies $B\models \phi$.
\end{proof}

Now let us find a property in $\EM(\EkI)$ corresponding to the existence of a Duplicator winning strategy in the $k$-round existential \EF game from $A$ to $B$. The meaning of having such a strategy is that for any tuple of elements $a_1,\dots,a_i\in A$, one can find a corresponding tuple $b_1,\dots,b_i\in B$ such that the substructures of $A$ and $B$ induced by the $a_j$'s and the $b_j$'s, respectively, are isomorphic. This means that any path embedding $P\emb \FI_k(A)$, with image $[a_1]\cvr \cdots\cvr [a_1,\dots,a_i]$, induces a path embedding $P\emb \FI_k(B)$ with image $[b_1]\cvr \cdots \cvr [b_1,\dots,b_i]$. This defines a cocone of embeddings with vertex $\FI_k(B)$ over the diagram of paths that embed into $\FI_k(A)$. Since the colimit of the latter diagram is $\FI_k(A)$ and the mediating morphisms $\FI_k(A)\to \FI_k(B)$ induced by such cocones of embeddings are precisely the pathwise embeddings, this leads to the following result:

\begin{proposition}
\label{p:ex-EF-struct}
The following statements are equivalent for all structures $A, B$:
\begin{enumerate}[label=(\roman*)]
\item \label{i:ex-EF-path}There exists a pathwise embedding $\FI_k(A)\rightarrow \FI_k(B)$ in $\EM(\EkI)$.
\item \label{i:ex-EF-win} Duplicator has a winning strategy in the $k$-round existential \EF game from $A$ to $B$.
\end{enumerate}
\end{proposition}

\begin{proof}
This follows from Theorem~\ref{th:ex-arbor} and Proposition~\ref{p:ex-arbor-EF}, so we only sketch the proof.

\ref{i:ex-EF-path} $\Rightarrow$ \ref{i:ex-EF-win}
Let $f\colon \FI_k(A)\rightarrow \FI_k(B)$ be a pathwise embedding in $\EM(\EkI)$. Then Duplicator has a winning strategy in the $k$-round existential \EF game from $A$ to $B$ defined as follows: if Spoiler has played elements $a_1,\dots,a_n\in A$, in the $n$-th round Duplicator responds with $\epsilon_{J(B)}(f([a_1,\dots,a_n]))\in J(B)$. 

\ref{i:ex-EF-win} $\Rightarrow$ \ref{i:ex-EF-path}
A Duplicator winning strategy for the $k$-round existential \EF game from $A$ to $B$ is in particular a winning strategy for the existential positive game, and so there is an arrow $f\colon \FI_k(A)\rightarrow \FI_k(B)$ in $\EM(\EkI)$ (equivalently, an arrow $F_k(A)\rightarrow F_k(B)$ in $\EM(\Ek)$) by Theorem~\ref{th:ep-Kleisli-arrows}. Explicitly, for any $[a_1,\dots,a_n]\in \FI_k(A)$, 
\[
f([a_1,\dots,a_n])=[h([a_1]),h([a_1,a_2]),\dots,h([a_1,\dots,a_n])]
\]
where $h([a_1,\dots,a_j])\in B$ is Duplicator's answer (according to their winning strategy) to Spoiler's play $[a_1,\dots,a_j]$. The winning condition for the existential game ensure that~$f$ is a pathwise embedding.
\end{proof}

Together, the two previous propositions yield the following characterisation of equivalence in the existential fragment of first-order logic with bounded quantifier rank:
\begin{corollary}
\label{cor:ex-EF-equiv}
Let $A,B$ be arbitrary structures.
Then $A\equiv_{\EFO_k}B$ if, and only if, there exist pathwise embeddings $\FI_k(A)\leftrightarrows \FI_k(B)$ in $\EM(\EkI)$.
\end{corollary}

\subsection{Existential pebble game}
The \emph{existential pebble game} from $A$ to $B$ is the variant of the pebble game in which Spoiler plays only in $A$ and Duplicator plays only in $B$. The winning condition for Duplicator remains unchanged and amounts to saying that the final positions of the pebbles form a partial isomorphism between $A$ and $B$.

\begin{restatable}{proposition}{propexpeblogic}
\label{p:ex-peb-logic}
The following statements are equivalent for all structures $A, B$:
\begin{enumerate}[label=(\roman*)]
\item\label{i:exist-pres-pebble} $A\IMP_{\ELL^{k}} B$, \ie every sentence of $\exists \LL^k$ that is true in $A$ is true in $B$.
\item\label{i:win-exis-peb} Duplicator has a winning strategy in the existential $k$-pebble game from $A$ to $B$.
\end{enumerate}
\end{restatable}

\begin{proof}
The proof is similar to that of Proposition~\ref{p:ex-EF-logic}; we outline the main ideas and leave the details to the appendix.

Given pebbles $p_1,\dots, p_n$ and elements $c_1,\dots,c_n$ of a structure $C$, we let $r_1,\dots,r_m\in \{1,\ldots,k\}$ be an enumeration of the $p_i$'s, with the $r_j$'s pairwise distinct, and denote by $l(r_j)$ the index of the last occurrence of pebble $r_j$ in $[(p_1,c_1),\dots,(p_n,c_n)]$.

\ref{i:exist-pres-pebble} $\Rightarrow$ \ref{i:win-exis-EF} The proof is the same, mutatis mutandis, as that of \ref{i:exist-pres-EF} $\Rightarrow$ \ref{i:win-exis-EF} in Proposition~\ref{p:ex-EF-logic}. The difference lies in the notion of type: given pebbles $p_1,\dots, p_n$ and elements $c_1,\dots,c_n\in C$, we define the \emph{$\exists \LL^k$-type} of $(p_1,c_1),\dots,(p_n,c_n)$ by 
\[
\tp_C((p_1,c_1),\dots,(p_n,c_n)) \coloneqq \{\psi(x_{r_1},\dots,x_{r_m})\in \exists \LL^k\mid C\models \psi(c_{l(r_1)},\dots,c_{l(r_m)})\}.
\] 

\ref{i:win-exis-EF} $\Rightarrow$ \ref{i:exist-pres-pebble} Suppose that $\phi\in \LL^k$ is such that $A\models \phi$. 
One can prove, by structural induction, that each subformula $\psi(x_{r_1},\dots,x_{r_m})$ of $\phi$ has the following property: Let $[(p_1,a_1),\dots,(p_n,a_n)]\in \FI_k(A)$ and, for all $i\in\{1,\dots,n\}$, let $b_i$ be Duplicator's response to Spoiler's play $[(p_1,a_1),\dots,(p_i,a_i)]$. 
Then 
\[
A\models \psi (a_{l(r_1)},\dots,a_{l(r_m)}) \ \Longrightarrow \ B\models \psi (b_{l(r_1)},\dots,b_{l(r_m)}).
\]
Taking $\psi=\phi$ yields the desired conclusion.
\end{proof}

As in the case of existential \EF games, the existence of a Duplicator's winning strategy is equivalent to the existence of a pathwise embedding between the relevant co-free coalgebras. We only sketch the proof, since this is a consequence of Theorem~\ref{th:ex-arbor} (in the axiomatic setting) and Proposition~\ref{p:ex-arbor-peb}.

\begin{proposition}
\label{p:ex-peb-struct}
The following statements are equivalent for all structures $A, B$:
\begin{enumerate}[label=(\roman*)]
\item\label{i:ex-pebble-path} There exists a pathwise embedding $\FI_k(A)\rightarrow \FI_k(B)$ in $\EM(\PkI)$.
\item\label{i:ex-pebble-win} Duplicator has a winning strategy in the existential $k$-pebble game from $A$ to $B$.
\end{enumerate}
\end{proposition}

\begin{proof}
\ref{i:ex-pebble-path} $\Rightarrow$ \ref{i:ex-pebble-win}
Suppose that $f\colon \FI_k(A)\rightarrow \FI_k(B)$ is a pathwise embedding in $\EM(\PkI)$. Duplicator has a winning strategy in the existential $k$-pebble game from $A$ to $B$ defined as follows: in the $n$-th round, if Spoiler's moves are given by $[(p_1,a_1),\dots,(p_n,a_n)]$, Duplicator responds by placing pebble $p_{n}$ on ${b_n\coloneqq \epsilon_{J(B)} (f([(p_1,a_1),\dots,(p_n,a_n)]))\in J(B)}$.

\ref{i:ex-pebble-win} $\Rightarrow$ \ref{i:ex-pebble-path}
Similar to the proof of \ref{i:ex-EF-win} $\Rightarrow$ \ref{i:ex-EF-path} in Proposition~\ref{p:ex-EF-struct}, a Duplicator winning strategy for the existential $k$-pebble game from $A$ to $B$ is also a winning strategy for the existential positive game. Thus, by Theorem~\ref{th:ep-Kleisli-arrows}, there is an arrow ${f\colon \FI_k(A)\rightarrow \FI_k(B)}$ in $\EM(\PkI)$; equivalently, an arrow ${F_k(A)\rightarrow F_k(B)}$ in $\EM(\Pk)$. We have 
\[
f([(p_1,a_1),\dots,(p_n,a_n)])=[(p_1,b_{1}),\dots,(p_n,b_{n})]
\]
where $b_{j}\coloneqq h([(p_1,a_1),\dots,(p_j,a_j)])\in B$ is Duplicator's answer (according to their winning strategy) to Spoiler's play $[(p_1,a_1),\dots,(p_j,a_j)]$. Again, one can show that the winning condition for the existential game ensures that~$f$ is a pathwise embedding.
\end{proof} 

Combining the previous two propositions, we obtain the following characterisation of equivalence in the existential fragment of $k$-variable (infinitary) logic:
\begin{corollary}
\label{cor:ex-peb-equiv}
Let $A,B$ be arbitrary structures.
Then $A\equiv_{\ELL^k}B$ if, and only if, there exist pathwise embeddings $\FI_k(A)\leftrightarrows \FI_k(B)$ in $\EM(\PkI)$.
\end{corollary}

\subsection{Existential bisimulation game}
\label{ss:ex-modal}

Analogously to before, define the \emph{existential bisimulation game} as the variant of the bisimulation game where Spoiler plays only in~$A$ and Duplicator plays in~$B$; the winning condition for Duplicator remains the same, \ie for all unary predicates $P$ and all pairs $(a,b)$ of selected elements, $P^{A}(a)\Leftrightarrow P^{B}(b)$.

The proof of the following result, which states that the existential bisimulation game captures the existential fragment of modal logic, is the same, mutatis mutandis, as that of Proposition~\ref{p:ex-EF-logic} and can be found in the appendix.

\begin{restatable}{proposition}{propexmodlogic}
\label{p:ex-mod-logic}
    The following statements are equivalent for all pointed Kripke models $(A,a)$ and $(B,b)$:
\begin{enumerate}[label=(\roman*)]
\item\label{i:EMLk-preserved} $(A,a)\IMP_{\EML_{k}} (B,b)$, \ie for all $\phi(x)\in \EML_k$, if $A\models \phi(a)$ then $B\models \phi (b)$.
\item\label{i:Dup-wins-k-ex-bisim} Duplicator has a winning strategy in the $k$-round existential bisimulation game from $(A,a)$ to $(B,b)$.
    \end{enumerate}
\end{restatable}

Similarly as before, the existence of a Duplicator winning strategy in the existential bisimulation game turns out to be equivalent to the existence of a pathwise embedding between the co-free coalgebras, which in this case are tree unravellings. We omit the proof of this fact as it is a consequence of Theorem~\ref{th:ex-arbor} and Proposition~\ref{p:ex-arbor-mod}.

\begin{proposition}
\label{p:ex-mod-struct}
    The following statements are equivalent for all pointed Kripke models $(A,a)$ and $(B,b)$: 
\begin{enumerate}[label=(\roman*)]
\item There exists a pathwise embedding $F_k(A,a)\rightarrow F_k(B,b)$ in $\EM(\Mk)$.
\item Duplicator has a winning strategy in the $k$-round existential bisimulation game from $(A,a)$ to $(B,b)$.
\end{enumerate}
    
\end{proposition}

Combining these two propositions yields a characterisation of equivalence in the existential fragment of modal logic with bounded depth:
\begin{corollary}
\label{cor:ex-modal-equiv}
Let $(A,a),(B,b)$ be arbitrary pointed Kripke models.
Then $A\equiv_{\EML_k}B$ if, and only if, there exist pathwise embeddings $F_k(A,a)\leftrightarrows F_k(B,b)$ in $\EM(\Mk)$.
\end{corollary}

%%%%%%%%%%%%%%%%%%%%%%%%%%%%%%%%%%%%%%%%%
%%%%%%%%%%%%%%%%%%%%%%%%%%%%%%%%%%%%%%%%%
\section{Positive fragments and positive bisimulations}
\label{s:concrete-pos}

For any fragment $\mathbb{L}\subseteq \LL_{\infty,\omega}$, the \emph{positive} fragment
\[
\prescript{+}{}{\mathbb{L}}
\] 
of $\mathbb{L}$ consists of the formulas that do not contain the negation symbol.

The purpose of this section, in a similar spirit to the previous one, is to give a comonadic characterisation of the preservation of the positive fragments $\PFO_k $, $\PLLk$ and $\PML_k$. This requires a notion of \emph{positive bisimulation}, the intuition of which is described in Section~\ref{ss:pos-EF} in the case of \EF games.

\begin{definition}\label{d:pos-bisim-concrete}
Let $\A$ be one of the categories $\EM(\Ek)$, $\EM(\Pk)$ or $\EM(\Mk)$. If $X,Y$ are objects of $\A$, a \emph{positive bisimulation} from $X$ to $Y$ is a diagram
\[
    \begin{tikzcd}
    Z_1\arrow{d}[swap]{p} \arrow{r}{h} & Z_2\arrow{d}{q} \\
    X & Y
    \end{tikzcd}
\]
in $\A$ such that $p,q$ are open pathwise embeddings and $h$ is a bijection. We write $X\pbisim Y$ to denote the existence of a positive bisimulation from $X$ to $Y$.
\end{definition}

Positive bisimulations can be regarded as spans by considering the morphisms $p$ and $q\circ h$. For a characterisation of these spans in the axiomatic setting, see Remark~\ref{rem:positive-bisim-as-spans}.

\begin{remark}
The difference between the notions of bisimulation (Definition~\ref{d:bisimulation}) and positive bisimulation lies in the fact that the arrow $h$ in Definition~\ref{d:pos-bisim-concrete} need not be an isomorphism, as it may not be an embedding of structures. Intuitively, we can think of~$Z_{1}$ and~$Z_{2}$ as having the same universe, with~$Z_{2}$ having ``more relations'' than~$Z_{1}$.
\end{remark}

Our second main result can be stated as follows:

\begin{theorem}\label{t:positive-concrete-charact}
The following statements hold for all (pointed) structures $A$ and $B$:
    \begin{enumerate}[label=(\arabic*)]
        \item $A\IMP_{\PFO_k}B $ if, and only if, $\FI_k(A) \pbisim \FI_k(B)$ in $\EM(\EkI)$.
        \item $A\IMP_{\PLLk}B $ if, and only if, $\FI_k(A) \pbisim \FI_k(B)$ in $\EM(\PkI)$.
        \item\label{i:positive-modal-char} $(A,a)\IMP_{\PML_k}(B,b) $ if, and only if, $F_k(A,a) \pbisim F_k(B,b)$ in $\EM(\Mk)$.
    \end{enumerate}
\end{theorem}

As in Section~\ref{s:concrete-exs}, for each of the three logic fragments, we establish this in two steps: 

\begin{enumerate}[label=(\roman*)]
\item first, we show that the positive fragments are captured by \emph{positive} variants of the corresponding games;
\item second, we prove that the existence of a Duplicator winning strategy in the positive game is equivalent to the existence of a positive bisimulation between the corresponding co-free coalgebras.
\end{enumerate}

Again, item~\ref{i:first-step} relies on a routine argument which is spelled out, for each of the three logic fragments, in the appendix. Item~\ref{i:second-step} is a special case of a general axiomatic result which will be established in Section~\ref{s:pos-arbor}. In the remainder of this section, we will provide a concrete intuition for item~\ref{i:second-step} in the case of \EF games, while in the other two cases we will state the main observations without proof and indicate how they can be deduced from later results.

%%%%%%%%%%%%%%%%%%%%%%%%%%%%%%%%%%%%%%%%%
%%%%%%%%%%%%%%%%%%%%%%%%%%%%%%%%%%%%%%%%%
\subsection{Positive \EF game}
\label{ss:pos-EF}
We consider the positive fragment $\PFO_k$ of $\FO_k$, consisting of formulas with quantifier rank at most $k$ that do not contain the negation symbol (but may contain both existential and universal quantifiers). 

The intuition is to define a game dual to the one in Section~\ref{s:concrete-exs}: the \emph{positive \EF game} from $A$ to $B$ is the variant of the \EF game where Spoiler can play in either $A$ or $B$, Duplicator answers in the other structure, and the winning condition for Duplicator, after $k$ rounds, is that the ensuing relation given by the $k$ pairs of chosen elements is a partial homomorphism from $A$ to $B$.
Indeed, we have

\begin{restatable}{proposition}{propposEFlogic}
\label{p:pos-EF-logic}
The following conditions are equivalent: 
\begin{enumerate}[label=(\roman*)]
\item\label{i:PFOk-preserv} $A\IMP_{\PFO_k}B $, \ie every sentence of $\PFO_k$ that is true in $A$ is true in $B$.
\item\label{i:Dupl-wins-positive-k-game} Duplicator has a winning strategy in the positive $k$-round \EF game from $A$ to $B$.
\end{enumerate}
\end{restatable}

The proof of Proposition~\ref{p:pos-EF-logic} uses the same strategy as the proof of Proposition~\ref{p:ex-EF-logic}, and shows that occurrences of the quantifier $\forall$ correspond to Spoiler's plays in~$B$. A detailed proof of this result can be found in the appendix.

The step into the structural world is more delicate than in the existential case. To illustrate the intuition, we need to recall the case of the full fragment $\FO_k$. In this case, the existence of a Duplicator's winning strategy amounts to saying that, given two isomorphic paths in $\FI_k(A)$ and $\FI_k(B)$, any extension of one of these paths in $A$ or in $B$ can be extended in the other structure while preserving the isomorphism.
Since the two paths are isomorphic, they can be regarded as one and the same path $P$. Thus, the situation can be represented by a span of embeddings
\begin{center}
    \begin{tikzcd}
      \FI_k(A)  &  P\arrow[l,tail] \arrow[r,tail] & \FI_k(B)
    \end{tikzcd}
\end{center}
where any extension of the image of $P$ in either $\FI_k(A)$ or $\FI_k(B)$ can be extended in the other structure to yield a similar span. Taking the ``union'' of such paths, we get a span
\begin{center}
    \begin{tikzcd}
     \FI_k(A) &  Z\arrow[l] \arrow[r] & \FI_k(B)
         \end{tikzcd}
\end{center}
where the two arrows are open pathwise embeddings \cite{AS2021}. Such a span is called a \emph{bisimulation} between $\FI_k(A)$ and $\FI_k(B)$.

In the positive case, by contrast, the two paths are no longer isomorphic, but there is a bijective homomorphism $\o{h}$ from the first to the second:

\begin{center}
    \begin{tikzcd}
       P_1\arrow[d,tail]\arrow[r,"\o{h}"]& P_2\arrow[d,tail]& \\
        \FI_k(A) &  \FI_k(B)
    \end{tikzcd}
\end{center}
Taking the union of these paths this time leads to a positive bisimulation from $\FI_k(A)$ to $\FI_k(B)$ in the sense of Definition~\ref{d:pos-bisim-concrete}. Thus, 

\begin{proposition}\label{p:pos-EF-struct}
There exists a positive bisimulation from $\FI_k(A)$ to $\FI_k(B)$ in $\EM(\EkI)$ if, and only if, Duplicator has a winning strategy in the positive $k$-round \EF game from $A$ to $B$.
\end{proposition}

\begin{proof}
The statement is a consequence of Theorem~\ref{th:game-bisim}, in the axiomatic setting, and Proposition~\ref{p:pos-arbor-EF}. Nevertheless, we sketch a proof in this concrete setting because it provides useful insights for understanding the proofs of the axiomatic results.

Suppose there exists a positive bisimulation from $\FI_k(A)$ to $\FI_k(B)$ as displayed below.
\[
    \begin{tikzcd}
    Z_1\arrow{d}[swap]{p} \arrow{r}{h} & Z_2\arrow{d}{q} \\
    \FI_k(A) & \FI_k(B)
    \end{tikzcd}
\]
To construct a Duplicator's winning strategy we use the following easy observation, the proof of which is left to the reader. 
\begin{fact}\label{fact:chains}
Let $z_1\cvr \cdots\cvr z_i$ and $z'_1\cvr \cdots\cvr z'_i$, with $i\geq 0$, be chains in $Z_1$ and $Z_2$, respectively, such that $h (z_j)=z'_j$. For any $s_{i+1}$ (resp.\ $t_{i+1}$) such that $p(z_i)\cvr s_{i+1}$ (resp.\ $q(z_i)\cvr  t_{i+1}$), there exist $z_{i+1}$ and $z'_{i+1}$ such that $z_i\cvr z_{i+1}$ and $z'_{i}\cvr z'_{i+1}$ (if $i=0$, $z_1$ and $z'_{1}$ are roots), with $z'_{i+1}=h(z_{i+1})$ and $p(z_{i+1})=s_{i+1}$ (resp.\ $q(z'_{i+1})=t_{i+1}$).
\end{fact}

Assume that after $i$ rounds the chosen elements are $[a_1,\dots,a_i],[b_1,\dots,b_i]$ inducing, inductively, chains $z_1\cvr \cdots\cvr z_i$ and $z'_1\cvr \cdots\cvr z'_i$ in $Z_1$ and $Z_2$, respectively, such that for all $j$ we have $h (z_j)=z'_j$,  $p(z_j)=[a_1,\dots,a_j]$ and $q(z'_j)=[b_1,\dots,b_j]$. We define Duplicator's strategy as follows:
\begin{itemize}
    \item If Spoiler plays $a_{i+1}$, consider the sequences $s_j=[a_1,\dots,a_j]$ and use Fact~\ref{fact:chains} to build $z_{i+1},z'_{i+1}$. Duplicator's answer is $b_{i+1}\coloneqq \epsilon_Bq(z'_{i+1})$.
    
    \item If Spoiler plays $b_{i+1}$, consider the sequences $t_j=[b_1,\dots,b_j]$ and use Fact~\ref{fact:chains} to build $z_{i+1},z'_{i+1}$. Duplicator's answer is $a_{i+1}\coloneqq \epsilon_A p(z_{i+1})$.
\end{itemize}

This strategy is well defined; let us prove that it is a winning strategy for Duplicator. We define $z_j,z'_j$ and $s_j,t_j$ as above for $j\leq k$.

If $a_i=a_j$ with $i<j$, then $I^{\FI_k(A)}(s_i,s_j)$. Since $p$ is a pathwise embedding it reflects relations on the path $s_1\cvr \cdots \cvr s_k$. Thus $I^{Z_1}(z_i,z_j)$. But $q\circ h$ is a homomorphism, so $I^{\FI_k(B)} (t_i,t_j)$ and thus $b_i=b_j$. This shows that relation defined by the chosen elements is functional.
A similar argument shows that $R^A(a_{i_1},\dots,a_{i_n})$ implies $R^B(b_{i_1},\dots,b_{i_n})$ for any relation symbol $R$, and so Duplicator's winning condition is satisfied.

Conversely, suppose that Duplicator has a winning strategy in the positive $k$-round \EF game from $A$ to $B$. Let $W(A,B)\subseteq A^{\leq k}\times B^{\leq k}$ consist of the pairs $(s,t)$ that are valid plays in the \EF game where Duplicator plays according to their winning strategy.
We define $\sg^{I}$-structures $Z_1$ and $Z_2$ on the universe $W(A,B)$ as follows: for every relation symbol $R$,
\[
R^{Z_1}((s_1,t_1),\dots,(s_n,t_n))\Leftrightarrow (\text{the $t_i$'s are pairwise comparable and } R^{\FI_k(A)}(s_1,\dots,s_n))
\]
and 
\[
R^{Z_2}((s_1,t_1),\dots,(s_n,t_n))\Leftrightarrow (\text{the $s_i$'s are pairwise comparable and } R^{\FI_k(B)}(t_1,\dots,t_n)).
\]
Upon equipping $Z_1$ and $Z_2$ with the product order, they can be regarded as objects of $\EM(\EkI)$, \ie they satisfy condition~\ref{i:E-bis}.

Now, define $h\colon Z_1\rightarrow Z_2$ to be the identity function of $W(A,B)$, and let $p\colon Z_1\rightarrow \FI_k(A)$ and $q\colon Z_2\rightarrow \FI_k(B)$ be the two projections from $W(A,B)$. 
It can be shown that $h$ is a bijective $\EM(\EkI)$-morphism, and that $p$ and $q$ are open pathwise embeddings.
\end{proof}

Combining Propositions~\ref{p:pos-EF-logic} and~\ref{p:pos-EF-struct} we get a comonadic characterisation of equivalence in the positive fragment of first-order logic with bounded quantifier rank:
\begin{corollary}
\label{cor:pos-EF-equiv}
The following conditions are equivalent for all structures $A,B$:
\begin{enumerate}[label=(\roman*)]
\item $A\equiv_{\PFO_k}B$.

\item There exist morphisms

\begin{center}
    \begin{tikzcd}
       &  Z\arrow[ld,"p",swap]& \\
        \FI_k(A) \arrow[rd,leftarrow,"q'"] & & \FI_k(B) \arrow[ul,leftarrow,"q",swap]\\
        & Z' \arrow[ru,"p'"]
    \end{tikzcd}
\end{center}
in $\EM(\EkI)$ such that $p,p'$ are open pathwise embeddings and $q,q'$ are compositions of a bijection followed by an open pathwise embedding.\footnote{The morphisms that are compositions of a bijection followed by an open pathwise embedding are exactly the tree-open morphisms in the axiomatic setting; cf.\ Definition~\ref{def:tree-open} and Remark~\ref{rem:positive-bisim-as-spans}. A similar remark applies to Corollaries~\ref{cor:pos-peb-equiv} and~\ref{cor:pos-mod-equiv}.}
\end{enumerate}
\end{corollary}

\begin{remark}
We will see in Section~\ref{s:arboreal} that the morphisms that can be obtained by composing a bijection with an open pathwise embedding admit a simple characterisation: they are precisely those whose underlying forest morphism is a surjective p-morphism (cf.\ Lemma~\ref{l:tree-open-charact}).
\end{remark}

\subsection{Positive pebble game}
The \emph{positive $k$-pebble game} from $A$ to $B$ is the variant of the $k$-pebble game in which Spoiler can play in either $A$ or $B$, Duplicator answers in the other structure, and the winning condition for Duplicator is that, at each round, the current positions of the pebbles define a partial homomorphism from $A$ to $B$.

\begin{restatable}{proposition}{proppospeblogic}
\label{p:pos-peb-logic}
The following conditions are equivalent:
\begin{enumerate}[label=(\roman*)]
\item\label{i:PLLK-preserved} $A\IMP_{\PLLk}B $, \ie every sentence of $\PLLk$ that is true in $A$ is true in $B$.
\item\label{i:Duplicator-wins-positive-k-peb} Duplicator has a winning strategy in the positive $k$-pebble game from $A$ to $B$.
\end{enumerate}
\end{restatable}

The proof of the previous result is similar to that of Proposition~\ref{p:pos-EF-logic} and can be found in the appendix.
We record here the following facts, akin to Proposition~\ref{p:pos-EF-struct} and Corollary~\ref{cor:pos-EF-equiv}, respectively, which will follow from Theorem~\ref{th:game-bisim} (in the axiomatic setting) combined with Proposition~\ref{p:pos-abor-peb}.

\begin{proposition}\label{p:pos-peb-struct}
There exists a positive bisimulation from $\FI_k(A)$ to $\FI_k(B)$ in $\EM(\PkI)$ if, and only if, Duplicator has a winning strategy in the positive $k$-pebble game from $A$ to $B$.
\end{proposition}

\begin{corollary}
\label{cor:pos-peb-equiv}
The following conditions are equivalent:
\begin{enumerate}[label=(\roman*)]
\item $A\equiv_{\PLLk}B$.
\item There exist morphisms
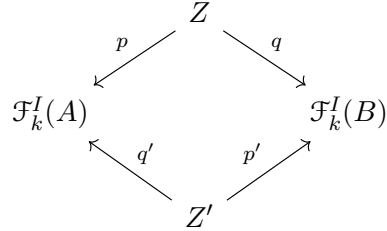
\begin{center}
    \begin{tikzcd}
       &  Z\arrow[ld,"p",swap]& \\
        \FI_k(A) \arrow[rd,leftarrow,"q'"] & & \FI_k(B) \arrow[ul,leftarrow,"q",swap]\\
        & Z' \arrow[ru,"p'"]
    \end{tikzcd}
\end{center}
in $\EM(\PkI)$ such that $p,p'$ are open pathwise embeddings and $q,q'$ are compositions of a bijection followed by an open pathwise embedding. 
\end{enumerate}
\end{corollary}

\subsection{Positive bisimulation game}
The \emph{positive bisimulation game} from $(A,a)$ to $(B,b)$ is the variant of the bisimulation game between $(A,a)$ and $(B,b)$ where Spoiler can play in either $A$ or $B$, Duplicator responds in the other structure, and the winning condition for Duplicator after $k$ rounds is given by weakening the condition in eq.~\eqref{eq:winning-bisim} to $P^A(a_i) \Rightarrow P^B(b_i)$. As before, we have the following fact (see the appendix for a proof).

\begin{restatable}{proposition}{propposmodlogic}
\label{p:pos-mod-logic}
The following two are equivalent: 
\begin{enumerate}[label=(\roman*)]
\item\label{i:PMLk-preserved} $(A,a)\IMP_{\PML_{k}} (B,b)$, \ie for all $\phi(x)\in \PML_k$, if $A\models \phi(a)$ then $B\models \phi (b)$.
\item\label{i:Dupl-wins-k-bisim-game} Duplicator has a winning strategy in the $k$-round bisimulation game from $(A,a)$ to $(B,b)$.
\end{enumerate}
\end{restatable}

Again, there is a structural characterisation which is a consequence of Theorem~\ref{th:game-bisim}.

\begin{proposition}    
\label{p:pos-mod-struct}
There exists a positive bisimulation from $F_k (A,a)$ to $F_k (B,b)$ in $\EM(\Mk)$ if, and only if, Duplicator has a winning strategy in the positive $k$-round bisimulation game from $(A,a)$ to $(B,b)$.
\end{proposition}

\begin{remark}
Positive bisimulations between Kripke models have been studied e.g.\ by Kurtonina and De Rijke in~\cite{Kurtonina97} under the name of \emph{directed simulations}. Combining Propositions~\ref{p:pos-mod-logic} and~\ref{p:pos-mod-struct} above, we deduce that $(A,a)\IMP_{\PML_{k}} (B,b)$ if and only if there is a positive bisimulation  from $F_k (A,a)$ to $F_k (B,b)$. A variant of this for image-finite Kripke models, and with no bound on the modal depth of formulas, was proved in~\cite[Proposition~3.4]{Kurtonina97}.
\end{remark}

\begin{corollary}
\label{cor:pos-mod-equiv}
The following conditions are equivalent:
\begin{enumerate}[label=(\roman*)]
\item $(A,a)\equiv_{\PML_{k}} (B,b)$.
\item There exist morphisms
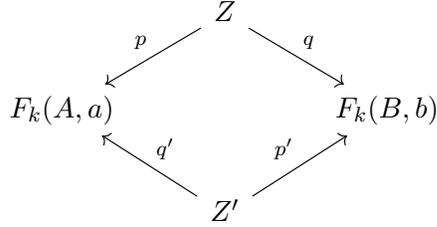
\begin{center}
    \begin{tikzcd}
       &  Z\arrow[ld,"p",swap]& \\
        F_k(A,a) \arrow[rd,leftarrow,"q'"] & & F_k(B,b) \arrow[ul,leftarrow,"q",swap]\\
        & Z' \arrow[ru,"p'"]
    \end{tikzcd}
\end{center}
in $\EM(\Mk)$ such that $p,p'$ are surjective p-morphisms (in the sense of modal logic, cf.\ Remark~\ref{rem:ope-bisim-modal}) and $q,q'$ are compositions of a bijection followed by a surjective p-morphism. 
\end{enumerate}
\end{corollary}

%%%%%%%%%%%%%%%%%%%%%%%%%%%%%%%%%%%%%%%%%
%%%%%%%%%%%%%%%%%%%%%%%%%%%%%%%%%%%%%%%%%
\section{Arboreal categories}
\label{s:arboreal}
%%%%%%%%%%%%%%%%%%%%%%%%%%%%%%%%%%%%%%%%%
%%%%%%%%%%%%%%%%%%%%%%%%%%%%%%%%%%%%%%%%%

We recall from \cite{AR2021icalp,AR2022} the main definitions and facts concerning the axiomatic approach to game comonads based on \emph{arboreal categories}. 

\subsection{Proper factorisation systems and paths}
Given arrows $e$ and $m$ in a category~$\C$, we say that $e$ has the \emph{left lifting property} with respect to $m$, or that $m$ has the \emph{right lifting property} with respect to $e$, and write $e {\,\pit\,} m$, if every commutative square
\[\begin{tikzcd}
{\cdot} \arrow{d}[swap]{e} \arrow{r} & {\cdot} \arrow{d}{m} \arrow[leftarrow, dashed]{dl} \\
{\cdot} \arrow{r} & {\cdot}
\end{tikzcd}\]
admits a diagonal filler (cf.\ Definition~\ref{d:open}). For any class $\mathscr{H}$ of morphisms in $\C$, let~${}^{\pit}\mathscr{H}$ (respectively, $\mathscr{H}^{\pit}$) be the class of morphisms having the left (respectively, right) lifting property with respect to each morphism in $\mathscr{H}$.

\begin{definition}\label{def:f-s}
A \emph{proper factorisation system} in a category $\C$ is a pair of classes of morphisms $(\Q,\M)$ satisfying the following conditions:
\begin{enumerate}[label=(\roman*)]
\item Each morphism $f$ in $\C$ can be decomposed as $f = m \circ e$ with $e\in \Q$ and $m\in \M$.
\item $\Q={}^{\pit}\M$ and $\M=\Q^{\pit}$.
\item Every arrow in $\Q$ is an epimorphism and every arrow in $\M$ a monomorphism.
\end{enumerate}
We refer to $\M$-morphisms as \emph{embeddings} and denote them by $\emb$. $\Q$-morphisms will be referred to as \emph{quotients} and denoted by~$\epi$.
\end{definition}

We state some well known properties of proper factorisation systems (cf.\ \cite{freyd1972categories} or~\cite{riehl2008factorization}): 
\begin{lemma}\label{l:factorisation-properties}
Let $(\Q,\M)$ be a proper factorisation system in $\C$. The following statements~hold:
\begin{enumerate}[label=(\arabic*)]
\item\label{compositions} $\Q$ and $\M$ are closed under compositions.
\item\label{isos} $\Q\cap\M=\{\text{isomorphisms}\}$.
\item\label{cancellation-e} $g\circ f\in \Q$ implies $g\in\Q$.
\item\label{cancellation-m} $g\circ f\in\M$ implies $f\in\M$.
\end{enumerate}
\end{lemma}

Let $\C$ be a category endowed with a proper factorisation system $(\Q,\M)$.
In the same way that one usually defines the poset of subobjects of a given object $X\in\C$, we can define the poset $\Emb{X}$ of \emph{$\M$-subobjects} of $X$, whose elements are equivalence classes~$[m]$ of embeddings $m\colon S\emb X$ (whenever convenient, we denote an equivalence class $[m]$ by any of its representatives).
For any morphism $f\colon X\to Y$ and embedding $m\colon S\emb X$, we can consider the $(\Q,\M)$-factorisation  of $f\circ m$: 
\[\begin{tikzcd}
S\arrow[twoheadrightarrow]{r} & \exists_f S \arrow[rightarrowtail]{r}{\exists_{f}m} & Y.
\end{tikzcd}\] 
This yields a monotone map $\exists_f\colon \Emb{X}\to\Emb{Y}$ sending $m$ to $\exists_{f}m$, and in fact it defines a functor from $\C$ to the category of posets and monotone maps.

\begin{definition}
An object $X$ of $\C$ is called a \emph{path} provided the poset $\Emb{X}$ is a finite linear order. Paths will be denoted by $P,Q$ and variations thereof. A \emph{path embedding} is an embedding $P\emb X$ whose domain is a path. 
\end{definition}

\begin{example}\label{ex:paths-EM-categories}
In the categories $\EM(\Ek)$, $\EM(\Pk)$ and $\EM(\Mk)$, consider the proper factorisation systems $(\Q,\M)$ where~$\Q$ consists of the surjective homomorphisms, 
and~$\M$ consists of the embeddings in the sense of Definition~\ref{d:pathwise}\ref{i:embedding}. In each case, the axiomatic notion of path coincides with the concrete one (see Definition~\ref{d:pathwise}\ref{i:path}).
\end{example}

\subsection{Arboreal categories defined}
Now that we have an abstract notion of path, we can define path embeddings, pathwise embeddings, bisimulations, and so forth, by generalising in a straightforward manner the definitions given in the concrete setting. In particular, extending Definition~\ref{d:bisimulation}, we can define a notion of \emph{bisimilarity} in any category equipped with a proper factorisation system.

The next step is to identify sufficient conditions on the category so that this notion of bisimilarity is well-behaved. In this regard, the sub-poset 
\[
\Path{X}
\] 
of $\Emb{X}$ consisting of the path embeddings plays a central role. In order for the maps $\exists_f\colon \Emb{X}\to\Emb{Y}$ to descend to maps $\Path{X}\to\Path{Y}$, we need that the quotient of a path is again a path. This in turn requires that the factorisation system $(\Q,\M)$ is \emph{stable}, \ie for every ${e\in \Q}$ and ${m\in\M}$ with common codomain, the pullback of~$e$ along~$m$ exists and belongs to $\Q$. Together with three additional axioms ensuring that the collection of paths behaves as expected, this leads to the notion of an arboreal category:

\begin{definition}\label{def:arboreal-cat}
An \emph{arboreal category} is a category $\A$, equipped with a stable proper factorisation system, that satisfies the following conditions:
\begin{description}\itemsep4pt
\item[Paths are connected]\label{ax:connected} Coproducts of paths exist and each path $P$ is \emph{connected}, \ie every arrow $P\to \coprod_{i\in I}{Q_i}$ into a coproduct of paths (with $I\neq\emptyset$) factors through some coproduct injection. 

\item[2-out-of-3 property]\label{ax:2-out-of-3} Given arrows $f\colon P\to Q$ and $g\colon Q\to Q'$ between paths, if any two of $f$, $g$ and $g\circ f$ are quotients, then so is the third. 

\item[Path-generation]\label{i:path-gen} Each object is \emph{path-generated}, \ie it is the colimit of the paths that embed into it.
\end{description}
\end{definition}

\begin{remark}
In Definition~\ref{def:arboreal-cat}, in the presence of the other two properties, the path-generation property is equivalent to saying that the full subcategory of $\A$ defined by the paths is \emph{dense}; see \cite[Theorem~5.1]{AR2022}.
\end{remark}

\begin{example}\label{ex:arboreal-cats}
The categories of coalgebras $\EM(\Ek)$, $\EM(\Pk)$ and $\EM(\Mk)$, equipped with the factorisation systems defined in Example~\ref{ex:paths-EM-categories}, are arboreal. \end{example}

We recall from~\cite{AR2022} some useful properties of paths that we will need later.
\begin{lemma}
\label{l:arboreal-facts}
The following statements hold for any object $X$ is an arboreal category:
\begin{enumerate}[label=(\arabic*)]
\item\label{i:quo-emb-path} If $X$ a path and $f\colon X \epi Y$, then $Y$ is a path. If $Y$ is a path and $f\colon X \emb Y$, then $X$ is a path.
\item\label{i:completely-join-irr} Any subset $U\subseteq \Path X$ admits a supremum $ \bigvee{U}$ in $\Emb{X}$. Moreover, if $U\neq\emptyset$ and $n\in \Path X$ satisfies $n\leq \bigvee{U}$, then $n\leq m$ for some $m\in U$. 
\item\label{i:one-emb-paths} For any two paths $P,Q$, there is at most one embedding $P\emb Q$.
\item\label{i:complete-lattices-of-subs} The poset $\Emb{X}$ is a complete lattice.
\end{enumerate}
\end{lemma}

\begin{proof}
Item~\ref{i:quo-emb-path} is \cite[Lemma~3.5]{AR2022}, while items~\ref{i:completely-join-irr} and~\ref{i:one-emb-paths} are direct consequences of \cite[Lemma~3.15 and Proposition~5.6]{AR2022}. For item~\ref{i:complete-lattices-of-subs}, see \cite[Remark~5.7]{AR2022}.
\end{proof}

\subsection{The functor $\Path$ of paths}

For any arboreal category $\A$, the assignment 
\[
(f\colon X\to Y) \ \longmapsto \ (\exists_{f}\colon \Path{X}\to \Path{Y})
\] 
yields a functor 
\[
\Path\colon \A\to\T
\] 
into the category $\T$ of non-empty trees and forest morphisms \cite[Theorem~3.11]{AR2022}. Note that, if $f$ is a pathwise embedding, we have $\Path f(m)  = f\circ m$ for all $m\in\Path{X}$.

In the remainder of this section, we collect some properties of the functor $\Path$ that we will need later. We begin with a preliminary observation.

\begin{lemma}\label{l:stable-emb}
The following statements hold in any arboreal category:
\begin{enumerate}[label=(\arabic*)]
\item\label{i:pullbacks-emb} The pullback of an embedding along an arbitrary arrow exists and is again a pullback.
\item\label{i:images-left-adj} For every arrow $f\colon X\to Y$, the monotone map $\exists_f\colon \Emb{X}\to\Emb{Y}$ is left adjoint.
\end{enumerate}
\end{lemma}

\begin{proof}
\ref{i:pullbacks-emb} Every morphism in an arboreal category can be decomposed as a quotient followed by an embedding, and the stability of the factorisation system ensures that pullbacks of embeddings along quotients exist. Thus it suffices to show that the pullback of an embedding along an embedding exists and is again an embedding. Note that the pullback of a diagram 
\[\begin{tikzcd}
S \arrow[rightarrowtail]{r} & X & T \arrow[rightarrowtail]{l}
\end{tikzcd}\]
can be computed as the infimum $S\wedge T$ in $\Emb{X}$, which exists because the latter poset is a (complete) lattice by Lemma~\ref{l:arboreal-facts}\ref{i:complete-lattices-of-subs}. Given this description, it is immediate to see that embeddings are stable under pullbacks along embeddings.\footnote{More generally, given any \emph{weak} factorisation system, if the pullback of an embedding along another embedding exists, then it is again an embedding. Cf.\ \cite{freyd1972categories} or~\cite{riehl2008factorization}.}

\ref{i:height-quotient} The right adjoint to $\exists_{f}$ is the monotone map $\Emb{Y} \to \Emb{X}$ that sends an embedding $S\emb Y$ to its pullback along $f$, which exists and is an embedding by item~\ref{i:pullbacks-emb}. The proof is a routine argument and is well known for posets of (strong) subobjects; in our setting, the same proof applies as in~\cite[Lemma~2.6]{AR2022}.
\end{proof}

For the next lemma, let us say that the \emph{height} of a path $P$ in an arboreal category, written $\htf(P)$, is the length of the chain $\Path{P}$.

\begin{lemma}\label{l:quotients-arb}
The following statements hold in any arboreal category:
\begin{enumerate}[label=(\arabic*)]
\item\label{i:quotient-surj} A morphism $f$ is a quotient if, and only if, $\Path{f}$ is surjective.
\item\label{i:height-quotient} A morphism $h\colon P\to Q$ between paths is a quotient if and only if $\htf(P)=\htf(Q)$.
\end{enumerate}
\end{lemma}
\begin{proof}
\ref{i:quotient-surj} Consider an arbitrary morphism $f\colon X\to Y$. The ``only if'' direction is the content of \cite[Lemma~5.5(b)]{AR2022}. For the reverse implication, suppose that $\Path{f}$ is surjective. Recall from \cite[Lemma~5.5(a)]{AR2022} that, for all objects $W$ in an arboreal category and all $m\in \Emb{W}$, 
\[
m = \bigvee{\{p\in \Path{W}\mid p\leq m\}}.
\]
In particular, $\id_{W}= \bigvee{\Path{W}}$.
Moreover, the monotone map~$\exists_{f}$ is left adjoint by Lemma~\ref{l:stable-emb}\ref{i:images-left-adj}; in particular, it preserves suprema.
Thus in the complete lattice $\Emb{Y}$ we have
\begin{align*}
\id_{Y} &= \bigvee{\Path{Y}} \\
&= \bigvee{\{\exists_{f}p \mid p\in \Path{X}\}} \tag*{$\Path{f}$ is surjective}\\
&= \exists_{f}(\bigvee{\Path{X}}) \tag*{$\exists_{f}$ is left adjoint}\\
&= \exists_{f}(\id_{X}),
\end{align*}
which in turn is equivalent to saying that $f$ is a quotient.

\ref{i:height-quotient} The map $\Path{h}$ is a forest morphism between finite chains, hence it is surjective (equivalently, bijective) if, and only if, $\htf(P)=\htf(Q)$. The statement then follows from the previous item.
\end{proof}

In the case of the arboreal categories in Example~\ref{ex:arboreal-cats}, the functor $\Path$ is faithful (in fact, it is naturally isomorphic to the forgetful functor that forgets the relational structure and retains only the forest order). The next lemma characterises when the functor $\Path$ is faithful, for any arboreal category.

\begin{lemma}\label{l:one-arrow-path}
The following are equivalent in any arboreal category $\A$:
\begin{enumerate}[label=(\roman*)]
    \item\label{i:P-faithful} The functor $\Path\colon \A\to\T$ is faithful.
    \item\label{i:one-arrow} For any two paths $P,Q$, there is at most one arrow $P\to Q$.
    \item\label{i:one-quotient} For any two paths $P,Q$, there is at most one quotient $P\epi Q$.
\end{enumerate}
\end{lemma}

\begin{proof}
    \ref{i:P-faithful} $\Rightarrow$ \ref{i:one-arrow} follows at once from the fact that there is at most one forest morphism between any two finite chains, while \ref{i:one-arrow} $\Rightarrow$ \ref{i:one-quotient} is obvious.

    For \ref{i:one-quotient} $\Rightarrow$ \ref{i:P-faithful}, consider arrows $f,g\colon X\to Y$ in $\A$ such that $\Path{f} = \Path{g}$. Since $X$ is path-generated, in order to prove that $f=g$ it suffices to show that if $m\colon P\emb X$ is an arbitrary path embedding into $X$, we have $f\circ m =g\circ m$. But $\Path{f}(m) = \Path{g}(m)$ implies the existence of a diagram as displayed below,
    \[\begin{tikzcd}
    P \arrow[twoheadrightarrow,yshift=3pt]{r}{e} \arrow[twoheadrightarrow,yshift=-3pt]{r}[swap]{e'} & Q \arrow[rightarrowtail]{r}{n} & X
    \end{tikzcd}\]
    where $f\circ m = n\circ e$ and $g\circ m = n\circ e'$. If there is at most one quotient $P\epi Q$ we get $e=e'$ and so $f\circ m =g\circ m$, as desired.
\end{proof}

In the concrete setting, e.g.\ in the category $\EM(\Ek)$, we can factor any morphism $f\colon X\to Y$ of forest-ordered structures as 
\[\begin{tikzcd}
X \arrow{r}{e} & X^{\circ} \arrow{r}{g} & Y
\end{tikzcd}\] 
where $e$ is the identity function and $g$ is a pathwise embedding whose underlying function coincides with $f$. The object $X^{\circ}$ is obtained from $X$ by keeping the same forest order and adding tuples of related elements along the branches in order to make $g$ a pathwise embedding. We shall now see that an appropriate generalisation of this factorisation, which will be needed in Section~\ref{s:pos-arbor} to construct positive bisimulations, is available in any arboreal category that satisfies some additional assumptions (which hold in all examples of arboreal categories mentioned in this paper).

The following definition is introduced in a forthcoming joint paper by Tom\'a\v{s} Jakl and the third author, where it is shown that, under appropriate assumptions, the functor~$\Path$ is a fibration and has a right adjoint. Proposition~\ref{p:fact-quo-pw} below is a direct consequence of~$\Path$ being a fibration.

\begin{definition}
Let $\A$ be an arboreal category.
\begin{enumerate}[label=(\arabic*)]
\item A \emph{tree-indexed diagram of paths} in $\A$ is a functor $T\to \A$, with $T$ a non-empty tree, that factors through the inclusion $\A_{p}\hookrightarrow\A$ of the full subcategory of paths.
\item An object $X$ of $\A$ is \emph{tree-connected} if for every tree-indexed diagram of paths $F\colon T\to \A$ and every morphism 
\[
f\colon X\to \colim{F}
\]
there is a least $j\in T$ such that $f$ factors through the colimit map ${F(j)\to \colim{F}}$.
\end{enumerate}
\end{definition}

\begin{proposition}\label{p:fact-quo-pw}
Let $\A$ be an arboreal category admitting colimits of tree-indexed diagrams of paths and such that all paths in $\A$ are tree-connected.
Each arrow $X\to Y$ in $\A$ can be decomposed as
\[\begin{tikzcd}
X \arrow{r}{e} & X^{\circ} \arrow{r}{g} & Y
\end{tikzcd}\]
where $g$ is a pathwise embedding and $\Path e$ is an isomorphism (in particular, it follows from Lemma~\ref{l:quotients-arb}\ref{i:quotient-surj} that $e$ is a quotient).
\end{proposition}

The morphisms that admit a decomposition of the form $g\circ e$, with $g$ an \emph{open} pathwise embedding and $\Path{e}$ an isomorphism, admit a particularly simple characterisation. In order to state this (see Lemma~\ref{l:tree-open-charact} below), we introduce the following terminology:

\begin{definition}\label{def:tree-open}
A morphism $f$ in an arboreal category is said to be \emph{tree-open} if $\Path f$ is a p-morphism (cf.\ Example~\ref{ex:p-morphisms-open-maps} for the latter notion).
\end{definition}

\begin{remark}
In all our concrete examples, a morphism between coalgebras for a game comonad is tree-open exactly when it is a surjective p-morphism between the underlying forest orders. Note that, in any arboreal category, tree-open morphisms are quotients by Lemma~\ref{l:quotients-arb}\ref{i:quotient-surj}.
\end{remark}

With this terminology, \cite[Proposition~4.2]{AR2022} can be rephrased as follows:
\begin{lemma}\label{l:open-tree-open} 
A pathwise embedding is open if, and only if, it is tree-open.
\end{lemma}

\begin{remark}\label{rem:tree-open-path-lift-prop}
A morphism $f\colon X\to Y$ is tree-open precisely when it satisfies the following path lifting property: Any commutative square as on the left-hand side below can be completed to a commutative diagram as on the right-hand side:\footnote{Observe that the left triangle in the rightmost diagram commutes provided that the upper triangle and the square on the bottom right commute. Also, the arrow $P\to Q'$ is necessarily an embedding.}
\begin{center}
\begin{tikzcd}[column sep=1em, row sep=1em]
P \arrow[rightarrowtail]{rr} \arrow{dd} & & X \arrow{dd}{f} \\
& \phantom{Q'} & \\
Q \arrow[rightarrowtail]{rr} & & Y
\end{tikzcd}
\ \ \ \ \ \ \ 
\begin{tikzcd}[column sep=1em, row sep=1em]
P \arrow[rightarrowtail]{rr} \arrow{dd} \arrow{dr} & & X \arrow{dd}{f} \\
& Q' \arrow[rightarrowtail]{ur} \arrow[twoheadrightarrow]{dl} & \\
Q \arrow[rightarrowtail]{rr} & & Y
\end{tikzcd}
\end{center}
(Since we will not need this fact in the following, we omit its proof.)
\end{remark}

\begin{lemma}\label{l:tree-open-charact}
The following statements are equivalent for any morphism $f$ in an arboreal category satisfying the assumptions of Proposition~\ref{p:fact-quo-pw}:
\begin{enumerate}[label=(\roman*)]
\item\label{i:tree-open} $f$ is tree-open.
\item\label{i:tree-open-decomp} There exist morphisms $e$ and $g$ such that $g$ is an open pathwise embedding, $\Path e$ is an isomorphism, and $f = g\circ e$.
\end{enumerate}
\end{lemma}

\begin{proof}
\ref{i:tree-open} $\Rightarrow$ \ref{i:tree-open-decomp} By Proposition~\ref{p:fact-quo-pw}, we can write $f$ as $g\circ e$ where $g$ is a pathwise embedding and $\Path e$ is an isomorphism. Since $\Path f = \Path g \circ \Path e$, $f$ is tree-open if, and only if, $g$ is tree-open. In turn, by Lemma~\ref{l:open-tree-open}, $g$ is tree-open precisely when it is open.

\ref{i:tree-open-decomp} $\Rightarrow$ \ref{i:tree-open} It follows from Lemma~\ref{l:open-tree-open} that $\Path{g}$ is a p-morphism, and therefore the composite $\Path{f}=\Path{g}\circ \Path{e}$ is also a p-morphism. That is, $f$ is tree-open.
\end{proof}

%%%%%%%%%%%%%%%%%%%%%%%%%%%%%%%%%%%%%%%%%
%%%%%%%%%%%%%%%%%%%%%%%%%%%%%%%%%%%%%%%%%
\section{Existential arboreal games and pathwise embeddings}
\label{s:ex-arbor}
%%%%%%%%%%%%%%%%%%%%%%%%%%%%%%%%%%%%%%%%%
%%%%%%%%%%%%%%%%%%%%%%%%%%%%%%%%%%%%%%%%%

The purpose of this section is to unify the existential variants of the three games into an axiomatic setting in which we can prove that Duplicator has a winning strategy if, and only if, there exists a pathwise embedding in the corresponding category.
We recall from \cite{AR2022} the notion of existential arboreal game.

\begin{definition}\label{d:ex-arbor-game}
    Let $\A$ be an arboreal category, and let $X, Y$ be objects of $\A$. The \emph{existential arboreal game} $\exists\game(X,Y)$ played from $X$ to $Y$ is defined as follows. A winning position is a pair $(m,n) \in \Path X \times \Path Y $ such that $\dom(m)\cong\dom(n)$. 
    \begin{itemize}
    \item The initial position is $(\bot_X,\bot_Y)$, where $\bot_X\colon P\emb X $ and $\bot_Y\colon Q \emb Y$ are the roots of $\Path X $ and $\Path Y$, respectively. If $P\not \cong Q $, Duplicator loses the game.  
    \item At each round, if the current position is~$(m,n)$, Spoiler chooses $m'\in \Path{X}$ such that $m \cvr m'$ and Duplicator answers by selecting $n'\in\Path{Y}$ satisfying $n\cvr n'$. Duplicator wins the round if the ensuing position~$(m',n')$ is winning. Duplicator wins the game if they have a strategy that is winning after $t$ rounds for all $t\geq 0$.
    \end{itemize}
\end{definition}

\begin{theorem}[{\cite[Proposition 6.18]{AR2022}}]\label{th:ex-arbor}
    The following statements are equivalent for any arboreal category $\A$ and for any two objects $X,Y$ of $\A$:

    \begin{enumerate}[label=(\roman*)]
        \item Duplicator has a winning strategy in $\exists\game(X,Y)$.

        \item There exists a pathwise embedding $f\colon X \rightarrow Y$.
    \end{enumerate}
\end{theorem}

We claim that the previous result is a common generalisation of Propositions~\ref{p:ex-EF-struct}, \ref{p:ex-peb-struct} and~\ref{p:ex-mod-struct}, when applied to the arboreal categories $\EM(\EkI)$, $\EM(\PkI)$ and $\EM(\Mk)$ respectively. To show this, we need to prove that in each of the three cases the arboreal game is equivalent to the concrete one, in the sense that Duplicator has a winning strategy for one of them if and only if they have a winning strategy for the other. This is the purpose of the following three propositions.

\begin{proposition}\label{p:ex-arbor-EF}
    Duplicator has a winning strategy for the $k$-round existential \EF game from $A$ to $B$ if, and only if, they have a winning strategy for the existential arboreal game $\exists \game(\FI_k(A),\FI_k(B))$ in the category $\EM(\EkI)$.
\end{proposition}

\begin{proof}
Note that the forest obtained by removing the root of the tree $\Path (\FI_k(A))$ is isomorphic to the underlying forest of $\FI_k(A)$. Explicitly, if $m \colon P \emb \FI_k(A)$ is an element of $\Path (\FI_k(A))$ different from the root, and $P$ has underlying order $p_1 \cvr \cdots \cvr p_i$, then the corresponding element of $\FI_k(A)$ is $m(p_{i})$. Similarly for $B$, and so valid positions in the arboreal game $\exists\game(\FI_k(A),\FI_k(B))$ correspond precisely to valid positions the $k$-round \EF game from $A$ to $B$.

We show that the isomorphism above restricts to a bijection between (non-initial) winning positions of $\exists\game(\FI_k(A),\FI_k(B))$, and winning positions of the $k$-round \EF game. (Note that the initial object of $\EM(\EkI)$ is the empty structure, so the initial position of $\exists\game(\FI_k(A),\FI_k(B))$ is always winning.)

\begin{itemize}
\item Let $(m,n)$ be a valid position in $\exists\game(\FI_k(A),\FI_k(B))$. If $(m,n)$ is winning, then $\dom(m)\cong \dom(n)$ and so we can assume without loss of generality that $\dom(m)$ and $\dom(n)$ are one and the same path $P$ whose underlying order is ${x_1 \cvr \cdots \cvr x_i}$. Write $s_j \coloneqq m(x_j)$ and $t_j \coloneqq n(x_j)$ for all $1\leq j \leq i$. We claim that $(s_i,t_i)$ is a winning position in the \EF game. 
Recalling that $s_{i} = [\epsilon_A m(x_1),\dots,\epsilon_A m(x_i)]$, and similarly for $t_{i}$, we have 
\begin{align*}
\epsilon_A s_j = \epsilon_A s_l \ & \Longleftrightarrow \  I^{\FI_k(A)}(s_j,s_l) \\
& \Longleftrightarrow \ I^P(x_j,x_l) \tag*{$m$ embedding} \\
& \Longleftrightarrow \ I^{\FI_k(B)}(t_j,t_l) \tag*{$n$ embedding} \\
& \Longleftrightarrow \ \epsilon_B t_j = \epsilon_B t_l.
\end{align*} 
Hence the relation $\{(a_j,b_j)\}_{1\leq j \leq i}$ is functional and injective.
Similarly, by replacing $I$ with $R\in \sigma$, we see that it is a partial isomorphism. 

\item Conversely, suppose that $(m,n)$ is a valid position of $\exists\game(\FI_k(A),\FI_k(B))$ such that the corresponding position in the \EF game is winning. Then the domains of $m$ and $n$ have the same height, say they are of the form ${x_1 \cvr \cdots \cvr x_i}$ and ${y_1 \cvr \cdots \cvr y_i}$ respectively. The forest isomorphism $x_{j}\mapsto y_{j}$ is also an isomorphism of structures between $\dom(m)$ and $\dom(n)$: just observe that a tuple of elements in $\FI_k(A)$ is in the interpretation of a relation symbol precisely when the tuple formed by the last elements in $J(A)$ is in the relation, and similarly for $\FI_k(B)$. But $\{(\epsilon_A m(x_j),\epsilon_A m(y_j))\mid 1\leq j \leq i\}$ is a partial isomorphism by assumption, hence $(m,n)$ is a winning position in $\exists\game(\FI_k(A),\FI_k(B))$.
\end{itemize}

Thus, the isomorphism between $\Path (\FI_k(A))$, with the root removed, and the underlying forest of $\FI_k(A)$ allows us to translate Duplicator winning strategies in either game into winning strategies in the other game.
\end{proof}

\begin{proposition}\label{p:ex-arbor-peb}
    Duplicator has a winning strategy for the existential $k$-pebble game from $A$ to $B$ if, and only if, they have a winning strategy for the existential arboreal game $\exists \game(\FI_k(A),\FI_k(B))$ in the category $\EM(\PkI)$.
\end{proposition}

\begin{proof}
The proof is analogous to that of Proposition~\ref{p:ex-arbor-EF}: there is an isomorphism between the forest obtained from $\Path(\FI_k(A))$ by removing the root, and the underlying forest of $\FI_k(A)$. This isomorphism restricts to a bijection between (non-initial) winning positions of $\exists\game(\FI_k(A),\FI_k(B))$, and winning positions in the existential $k$-pebble game:  
\begin{itemize}
\item If $(m,n)$ is a winning position of $\exists\game(\FI_k(A),\FI_k(B))$ then $\dom(m)\cong \dom(n)$, so we can assume without loss of generality that $\dom(m)$ and $\dom(n)$ are one and the same path $Q$ with underlying order $y_1 \cvr \cdots \cvr y_n$. For all $i\in \{1,\ldots,n\}$, let $s_i = m(y_i)$ and $t_i = n(y_i)$. 
Morphisms in $\EM(\PkI)$ preserve pebble functions, hence for every $i$ we have $p_A(s_i) = p_Q(y_i) = p_B(t_i) \coloneqq p_i$. Since the pebble indices in the sequences $s_{n}$ and $t_{n}$ coincide, $(s_n,t_n)$ is a valid position in the existential $k$-pebble game.
Let us show that the latter is a winning position. Let $i,j\leq n$ be the last occurrences of pebbles $p_i,p_j$ in $p_1,\dots,p_n$, and suppose pebbles $p_{i},p_{j}$ are placed on elements $a_{i},a_{j}$ and $b_{i},b_{j}$ of $A$ and $B$, respectively. By assumption, given that the $s_i$ (respectively, the $t_i$) are pairwise comparable, and $m,n$ are embeddings, we get
\[
a_i = a_j \ \Leftrightarrow \ I^{\FI_k(A)}(s_i,s_j) \ \Leftrightarrow \ I^Q(y_i,y_j) \ \Leftrightarrow \ I^{\FI_k(B)}(t_i,t_j) \ \Leftrightarrow \ b_i = b_j.
\] 
The same reasoning applies, mutatis mutandis, to any relation symbol in the vocabulary in place of $I$, and to any subset of the $a_i$'s and $b_i$'s such that $i$ is the last occurrence of the pebble $p_i$. Therefore, the current positions of the pebbles form a partial isomorphism between $A$ and $B$.

\item Conversely, suppose $(m,n)$ is a valid position of $\exists\game(\FI_k(A),\FI_k(B))$ whose corresponding position $[(p_1,a_1),\dots,(p_n,a_n)], [(p_1,b_1),\dots,(p_n,b_n)]$ in the existential $k$-pebble game is winning. As in the proof of Lemma~\ref{p:ex-arbor-EF}, it can be shown that
\[
[(p_1,a_1),\dots,(p_i,a_i)] \mapsto [(p_1,b_1),\dots,(p_i,b_i)]
\] 
is an isomorphism in the category $\EM(\PkI)$ between $\dom(m)$ and $\dom(n)$, and so $(m,n)$ is a winning position. 
\end{itemize}

The translation of Duplicator winning strategies in one game to winning strategies in the other game is then obtained by applying the isomorphism between $\Path (\FI_k(A))$, with the root removed, and the underlying forest of $\FI_k(A)$.
\end{proof}

\begin{proposition}\label{p:ex-arbor-mod}
    Duplicator has a winning strategy for the $k$-round existential bisimulation game from $(A,a)$ to $(B,b)$ if, and only if, they have a winning strategy for the existential arboreal game $\exists \game(F_k (A,a), F_k (B,b))$ in the category $\EM(\Mk)$.
\end{proposition}

\begin{proof}
We spell out the translation between positions of the game $\exists \game(F_k (A,a), F_k (B,b))$ and positions of the $k$-round existential bisimulation game. For any pointed Kripke model $(C,c)$, the play in the existential bisimulation game corresponding to a path embedding $m\colon P \emb F_k(C,c)$, with $P$ given by $x_1 \cvr \cdots \cvr x_i$, is simply $[m(x_{1}),\ldots, m(x_i)]$. Conversely, a play 
\[
c=c_0 \xrightarrow{R_1} c_1 \to\cdots\xrightarrow{R_i}c_i
\]
of length at most $k$ in $(C,c)$ corresponds to the path embedding $m\colon P \emb F_k(C,c)$ where $P$ is the path $x_1 \cvr \cdots \cvr x_i$ and $m$ sends $x_{j}$ to 
\[
c_0 \cvr (c_0 \xrightarrow{R_1} c_1 ) \cvr (c_0 \xrightarrow{R_1} c_1 \xrightarrow{R_2} c_2 ) \cvr \cdots \cvr (c_0 \xrightarrow{R_1} c_1 \to\cdots\xrightarrow{R_j}c_j).
\]
These assignments define an isomorphism between the tree $\Path(F_k(C,c))$ and the underlying tree of $F_k(C,c)$. When applied to the Kripke models $(A,a)$ and $(B,b)$, it is not difficult to see that this yields a bijection between Duplicator winning strategies in $\exists \game(F_k (A,a), F_k (B,b))$ and Duplicator winning strategies in the $k$-round existential bisimulation game from $(A,a)$ to $(B,b)$.
\end{proof}

As the concrete definitions of pathwise embedding coincide with the arboreal one, the latter three propositions show that Theorem~\ref{th:ex-arbor} is indeed a common generalisation of the results stated in Section~\ref{s:concrete-exs}, when applied to the Eilenberg-Moore categories $\EM(\EkI)$, $\EM(\PkI)$ and $\EM(\Mk)$ respectively.

\section{Positive arboreal games and positive bisimulations}
\label{s:pos-arbor}

Following the same strategy as in the previous section, we will now unify the positive variants of the three model-comparison games, and relate them to an axiomatic notion of positive bisimulation. For the following definitions, we fix an arbitrary arboreal category $\A$ and objects $X,Y$ of $\A$. We write $X\to Y$ if there exists an arrow from $X$ to $Y$; if no such arrow exists, we write $X\not\to Y$.

\begin{definition}\label{d:pos-arbor-game}
The \emph{positive arboreal game} $\game^+(X,Y)$ played from $X$ to $Y$ is defined as follows. A winning position is a pair $(m,n) \in \Path X \times \Path Y $ such that $\dom(m)\to\dom(n)$. 
    \begin{itemize}
    \item The initial position is $(\bot_X,\bot_Y)$, where $\bot_X\colon P\emb X $ and $\bot_Y\colon Q \emb Y$ are the roots of $\Path X $ and $\Path Y$, respectively. If $P\not \rightarrow Q$, Duplicator loses the game.  
    \item At each round, if the current position is~$(m,n)$, either Spoiler chooses $m'\in \Path{X}$ such that $m \cvr m'$ and Duplicator responds with some $n'\in\Path{Y}$ satisfying $n\cvr n'$, or Spoiler chooses $n''$ with $n \cvr n''$ and Duplicator responds with $m''$ such that $m\cvr m''$. Duplicator wins the round if the ensuing position is winning, and they win the game if they have a strategy that is winning after $t$ rounds for all~$t\geq 0$.
    \end{itemize}
\end{definition}

\begin{remark}\label{r:winning-quotient}
    The definition of the game $\game^+(X,Y)$ ensures that any valid position $(m,n)$ satisfies $\htf(m)=\htf(n)$. Hence, by item~\ref{i:height-quotient} in Lemma~\ref{l:quotients-arb}, winning positions can equivalently be defined by requiring the existence of a quotient $\dom(m)\epi\dom(n)$.
\end{remark}

\begin{definition}\label{d:pos-bisim-arb}
    A \emph{positive bisimulation} from $X$ to $Y$ is a diagram
\begin{center}
    \begin{tikzcd}
    Z_1\arrow{d}[swap]{f} \arrow{r}{h} &Z_2\arrow{d}{g} \\
    X & Y
    \end{tikzcd}
    \end{center}
such that $f$ and $g$ are open pathwise embeddings and $\Path h$ is a tree isomorphism.
\end{definition}

\begin{remark}
In the concrete cases, the assumption that $\Path h$ is a tree isomorphism is equivalent to saying that $h$ is a bijective morphism of coalgebras. Therefore, Definition~\ref{d:pos-bisim-arb} is a direct generalisation of Definition~\ref{d:pos-bisim-concrete}.
\end{remark}

\begin{remark}\label{rem:positive-bisim-as-spans}
Positive bisimulations can be equivalently described as spans where one leg is an open pathwise embedding, and the other leg is the composition of an arrow~$h$ such that $\Path h$ is a tree isomorphism followed by an open pathwise embedding. In view of Lemma~\ref{l:tree-open-charact}, under mild assumptions (which are satisfied by all our concrete examples), positive bisimulations can be identified with spans of the form
\[\begin{tikzcd}
X  & Z \arrow{l}[swap]{f} \arrow{r}{g'} & Y
\end{tikzcd}\]
where $f$ is an open pathwise embedding and $g'$ is tree-open.
\end{remark}

The main result of this section is the following:
\begin{theorem}\label{th:game-bisim}
Let $\A$ be an arboreal category such that the functor $\Path\colon \A \to \T$ is faithful, $\A$ admits colimits of tree-index diagrams of paths, and paths in $\A$ are tree-connected. For any two objects $X,Y$ of $\A$ admitting a product, Duplicator has a winning strategy in $\game^+(X,Y)$ if and only if there exists a positive bisimulation from $X$ to $Y$.
\end{theorem}

The arboreal categories $\EM(\EkI)$, $\EM(\PkI)$ and $\EM(\Mk)$ are easily seen to satisfy the assumptions of the previous theorem. Moreover, we will see in Section~\ref{ss:pos-arbor-games} that the positive arboreal game played in the latter categories is equivalent to the positive versions of the $k$-round \EF game, $k$-pebble game, and $k$-round bisimulation game, respectively. Therefore, Propositions~\ref{p:pos-EF-struct}, \ref{p:pos-peb-struct} and~\ref{p:pos-mod-struct} in Section~\ref{s:concrete-pos} are instances of Theorem~\ref{th:game-bisim}.

To prove the latter result, we first rephrase positive arboreal games in terms of \emph{positive back-and-forth systems} in Section~\ref{ss:pos-backforth}. We then show in Section~\ref{ss:pos-arbor-bisim} that, under the assumptions of Theorem~\ref{th:game-bisim}, there exists a positive back-and-forth system from~$X$ to~$Y$ if and only if there exists a positive bisimulation from~$X$ to~$Y$.

\subsection{Positive back-and-forth-systems}\label{ss:pos-backforth}

The following is a positive variant of the notion of back-and-forth system in arboreal categories introduced in \cite[Definition~6.1]{AR2022}. Given objects $X$ and $Y$, we write $\brp{m,n}$ to denote a pair $(m,n)\in \Path X \times \Path Y$ such that there exists a quotient $\dom(m)\epi \dom(n)$.

\begin{definition}\label{d:pos-backforth}
A \emph{positive back-and-forth system} from $X$ to $Y$ is a subset $\B$ of $\Path X \times \Path Y$ consisting of pairs $\brp{m,n}$ such that:

    \begin{enumerate}[label=(\roman*)]
        \item\label{i:pos-back-forth-base} $\brp {\bot_X,\bot_Y} \in \B$.

        \item\label{i:pos-back-forth-forth} If $\brp{m,n}\in \B$ and $m\cvr m'$, there exists $n'$ such that $n\cvr n'$ and $\brp{m',n'}\in \B$.

        \item\label{i:pos-back-forth-back} If $\brp{m,n}\in \B$ and $n\cvr n'$, there exists $m'$ such that $m\cvr m'$ and $\brp{m',n'}\in \B$.
    \end{enumerate}
A positive back-and-forth system $\B$ is \emph{strong} if for all $\brp{m,n}\in \B$, if $m' \cvr m$ and $n'\cvr n$, then $\brp{m',n'}\in \B$.
\end{definition}

\begin{remark}
In the definition of strong positive back-and-forth system, the condition $\brp{m',n'}$ is automatically satisfied. Indeed, if $f\colon \dom(m)\epi \dom(n)$, $m'\cvr m$ and $n'\cvr n$, then the (quotient, embedding) decomposition of the restriction of $f$ to $\dom(m')$ yields a quotient $\dom(m')\epi \dom(n')$. Just observe that quotients between paths preserve the height, and $n'$ is the only element below $n$ of height $\htf(n)-1$.
\end{remark}

Given a positive back-and-forth system, we can always construct a strong positive back-and-forth system; this simple observation will be useful in the next section.

\begin{lemma}\label{l:strong-backforth}
    If there exists a positive back-and-forth system from $X$ to $Y$, there exists also a strong positive back-and-forth system from $X$ to $Y$.
\end{lemma}

\begin{proof}
This is a straightforward adaptation of the proof of \cite[Lemma~6.3]{AR2022}. Given a positive back-and-forth system $\B$ from $X$ to $Y$, we construct a strong positive back-and-forth system $\B'$ by removing those pairs that prevent $\B$ from being strong. 

Any pair $\brp{m,n}$ satisfies $\htf(m)=\htf(n)$; write $k$ for the latter natural number, and for every $j\in \{0,\ldots, k\}$ let $m_{j}$ and $n_{j}$ be the unique path embeddings of height $j$ below $m$ and $n$, respectively. We say that $\brp{m,n}$ is \emph{reachable} if $\brp{m,n}\in \B$ for all $j\in \{0,\ldots, k\}$. The subset $\B'$ of $\B$ consisting of the reachable pairs is easily seen to be a strong positive back-and-forth system from $X$ to $Y$. 
\end{proof}

The existence of a Duplicator winning strategy in $\game^+(X,Y)$ can then be translated in terms of positive back-and-forth systems:

\begin{proposition}\label{p:game-backforth}
Duplicator has a winning strategy in the game $\game^+(X,Y)$ if, and only if, there exists a (strong) positive back-and-forth system from $X$ to $Y$.
\end{proposition}

\begin{proof}
If Duplicator has a winning strategy for $\game^+(X,Y)$, we can inductively construct the set $\B$ of all positions of the game that can be reached if Duplicator plays according to their winning strategy. Using Remark~\ref{r:winning-quotient}, it is easy to see that this is a (strong) positive back-and-forth system from $X$ to $Y$.

Conversely, if there exists a positive back-and-forth system from $X$ to $Y$, condition~\ref{i:pos-back-forth-base} in Definition~\ref{d:pos-backforth} ensures that the initial position is winning, while conditions~\ref{i:pos-back-forth-forth} and~\ref{i:pos-back-forth-back} provide Duplicator with a way of responding to Spoiler's moves so that the ensuing position is again winning.
\end{proof}

\subsection{Positive bisimulations}
\label{ss:pos-arbor-bisim}
%%%%%%%%%%%%%%%%%%%%%%%%%%%%%%%%%%%%%%%%%
%%%%%%%%%%%%%%%%%%%%%%%%%%%%%%%%%%%%%%%%%

To establish Theorem~\ref{th:game-bisim}, it remains to prove that there exists a positive bisimulation from $X$ to $Y$ if and only if there exists a positive back-and-forth system from $X$ to $Y$. We start with the easier of the two implications.

\begin{proposition}
\label{p:bisim-backforth}
If there exists a positive bisimulation from $X$ to $Y$, then there exists a positive back-and-forth system from $X$ to $Y$.
\end{proposition}

\begin{proof}
Suppose there exists a positive bisimulation from $X$ to $Y$ as displayed below.

\begin{center}
    \begin{tikzcd}
    Z_1\arrow{d}[swap]{f} \arrow{r}{h} &Z_2\arrow{d}{g} \\
    X & Y
    \end{tikzcd}
    \end{center}
Note that $h$ is tree-open because $\Path h$ is an isomorphism of trees, and $f,g$ are tree-open by Lemma~\ref{l:open-tree-open}.
Consider a path embedding $m\colon P\emb Z_1$. Since $f$ and $g$ are pathwise embeddings, we have $\Path f(m)=f\circ m$ and $\Path(g\circ h)(m)=g\circ \Path h(m)$.
By definition of~$\Path$ on morphisms we have $\dom(m)\epi \dom(\Path h(m))$ and thus 
\[
\brp{\Path f(m), \Path (g\circ h)(m)}.
\]
Set $\B\coloneqq\{\brp{\Path f(m), \Path (g\circ h)(m)}\} \mid m\in \Path Z_1\}$.

We claim that $\B$ is a positive back-and-forth system from $X$ to $Y$.
\begin{enumerate}[label=(\roman*)]
    \item $\Path f$ and $\Path (g\circ h)$ preserve the roots, hence 
    \[
    \brp{\Path f(\bot_{Z_1}),\Path (g\circ h)(\bot_{Z_1})} =\brp{\bot_X,\bot_Y} \in \B.
    \]
    
    \item Let $\brp{m,n} \in \B$, and suppose that $m=f\circ p$ and $n=g\circ \Path h(p)$ for some $p\in \Path Z_1$.
If $m'\in \Path X$ is such that $m\cvr m'$, since $f$ is tree-open there is $e\in \Path Z_{1}$ such that $p\cvr e$ and $f\circ e= m'$. Then $n'\coloneqq g\circ \Path h(e)$ satisfies $n\cvr n'$ because $\Path (g\circ h)$ preserves $\cvr$, and moreover $\brp{m', n'} = (\Path f(e), \Path (g\circ h)(e)) \in \B$.
 
      \item As in the previous item, suppose $\brp{m,n} \in \B$, with $m=f\circ p$ and $n=g\circ \Path h(p)$ for some $p\in \Path Z_1$. If $n\in \Path Y$ satisfies $n\cvr n'$, since $g,h$ are tree-open there exist $e'\in \Path Z_{2}$ such that $\Path h (p) \cvr  e'$ and $g\circ e' = n'$, and $e\in \Path Z_{1}$ such that $p\cvr e$ and $\Path h(e) = e'$. As $\Path f$ preserves $\cvr$ we get $m\cvr m'$; furthermore, $\brp{m', n'} = (\Path f(e), \Path (g\circ h)(e)) \in \B$. \qedhere
\end{enumerate}
\end{proof}

The converse implication requires more work. To construct a positive bisimulation starting from a positive back-and-forth system, we need to understand what the objects $Z_1$ and $Z_2$ will be.
In the corresponding result for bisimulations, cf.\ Theorem~\ref{th:full-bisimulations}, the objects $Z_{1},Z_{2}$ coincide and their underlying set consists of the winning positions $W(A,B)\subseteq \FI_k (A) \times \FI_k(B)$, equipped with the relational structure induced by that of the product $\FI_k (A) \times \FI_k(B)$, which in our case corresponds to $X\times Y$. The same approach will not work for positive bisimulations, because $X$ might have ``fewer relations'' than~$Y$, which will result in $X\times Y$ also having fewer relations, and so there may be no subobject of $X\times Y$ that admits a pathwise embedding to $Y$. 

This leads to the idea of considering a factorisation of the projections $\pi_X\colon X\times Y \to X$ and $\pi_Y\colon X\times Y \to Y$ 
as $X\times Y \rightarrow X^0 \rightarrow X$ and $X\times Y \rightarrow Y^0 \rightarrow Y$, respectively, such that $X^0$ (respectively, $Y^0$) has ``the shape of $X\times Y$'' and ``the relational structure of~$X$ (respectively, of $Y$)''. The objects $Z_1$ and $Z_{2}$ will then  embed into $X^0$ and $Y^0$, respectively. The desired factorisation of the projections is provided by Proposition~\ref{p:fact-quo-pw}: the fact that $\Path X^0$ is isomorphic to $\Path (X\times Y)$ means that $X^0$ has the shape of $X\times Y$, and the fact that the arrow $X^0\rightarrow X$ is a pathwise embedding means that $X^{0}$ has the relational structure of $X$.

\begin{proposition}
\label{p:backforth-bisim}
Suppose that $\A$ satisfies the assumptions of Theorem~\ref{th:game-bisim}, and let $X,Y$ be objects of $\A$ admitting a product. If there exists a (strong) positive back-and-forth system from $X$ to $Y$, there exists a positive bisimulation from $X$ to $Y$.
\end{proposition}

\begin{proof}
Suppose there exists a positive back-and-forth system $\B$ from $X$ to $Y$. In view of Lemma~\ref{l:strong-backforth}, we can assume that $\B$ is strong.
We will construct a positive bisimulation
\begin{center}
    \begin{tikzcd}
    Z_1\arrow{d}[swap]{f} \arrow{r}{h} &Z_2\arrow{d}{g} \\
    X & Y
    \end{tikzcd}
    \end{center}
 from $X$ to $Y$ in four steps. The second and third steps follow closely the proof of \cite[Proposition 6.9]{AR2022}: we construct the suprema of two sets of path embeddings, and show that the restrictions of the appropriate projection maps to these subobjects are open pathwise embeddings. 

\vspace{0.5em}
\textbf{First step:} We start with an auxiliary construction. Consider the decompositions of the projections $\pi_{X}\colon X\times Y \to X$ and $\pi_{Y}\colon X\times Y \to Y$ given by Proposition~\ref{p:fact-quo-pw}: 
\[
X\times Y \xrightarrow{p_X} X^0 \xrightarrow{q_X} X, \ \ X\times Y \xrightarrow{p_Y}Y^0 \xrightarrow{q_Y} Y.
\] 
$\Path p_X, \Path p_Y$ are tree isomorphisms, and $q_X, q_Y$ are pathwise embeddings.

By Lemma~\ref{l:one-arrow-path}, for each $\brp{m,n} \in\B$ there is exactly one quotient $\lambda\colon \dom(m) \epi \dom(n)$. Write $P\coloneqq \dom(m)$, $Q\coloneqq\dom(n)$, and consider the arrow 
\[
\langle m,n\circ \lambda\rangle \colon P \rightarrow X\times Y.
\] 
The latter is a path embedding; just observe that $\pi_{X}\circ \langle m,n\circ \lambda\rangle = m$, which is an embedding. Define the path embeddings 
\[
\langle m,n\rangle_l\coloneqq \Path{p_{X}}(\langle m,n\circ \lambda\rangle)\colon P'\emb X^0
\]
and 
\[
\langle m,n\rangle_r\coloneqq \Path{p_{Y}}(\langle m,n\circ \lambda\rangle)\colon Q'\emb Y^0.
\]
Since $q_{X}\circ p_X\circ \langle m,n\circ \lambda \rangle = m$ and the latter is an embedding, it follows that $p_X\circ \langle m,n\circ \lambda \rangle$ is an embedding; hence $\langle m,n\rangle_l = p_X\circ \langle m,n\circ \lambda \rangle$. 

To improve readability, we will sometimes write $\hat{m}$ and $\hat{n}$ instead of $\langle m,n\rangle_l$ and $\langle m,n\rangle_r$, respectively. The uniqueness of $\lambda$ implies that $\hat{m},\hat{n}$ depend only on $m,n$.
Note that $P'\cong P$, thus the domain of $\hat{m}$ can be identified with the domain of $m$. Similarly, we claim that $Q'\cong Q$, so the domain of $\hat{n}$ can be identified with that of $n$.
Consider the following commutative square:
    \begin{center}
        \begin{tikzcd}
               P\arrow[d,two heads]\arrow[r,two heads,"\lambda"]& Q\arrow[d,tail,"n"]\\
               Q'\arrow[r,tail,"q_Y\circ \hat{n}"]&Y
        \end{tikzcd}
    \end{center}
 Since (quotient, embedding) factorisations are unique up to isomorphism, there is an isomorphism $Q\cong Q'$ making the diagram above commute. Henceforth we will identify $Q$ with $Q'$, so that 
    \begin{equation}\label{eq:nprime}
        p_Y\circ \langle m,n\circ \lambda\rangle=\langle m,n\rangle_r\circ \lambda.
    \end{equation}

We thus have a commutative diagram as displayed below:
\begin{equation}\label{eq:parallel-embeddings}
    \begin{tikzcd}
           & P \arrow[r,two heads,"\lambda"]\arrow[d,tail,"\hat{m}"] \arrow[ld,tail,"m" swap] & Q \arrow[d,tail,"\hat{n}" swap] \arrow[rd,tail,"n"]& \\
           X & X^0 \arrow[l,"q_X" swap] & Y^0 \arrow[r,"q_Y"] & Y
    \end{tikzcd}
\end{equation}

\vspace{0.5em}    
\textbf{Second step:} We now construct the objects $Z_1$ and $Z_2$, along with the pathwise embeddings $f$ and $g$.
Consider the sets
\[
U_1\coloneqq \{\langle m,n\rangle_l \mid \brp{m,n} \in \B\}\subseteq \Path X^0
\]
and
\[ 
U_2\coloneqq\{\langle m,n\rangle_r \mid \brp{m,n} \in \B\}\subseteq \Path Y^0,
\] 
which are downwards closed in $\Emb X^0$ and $\Emb Y^0$, respectively, because $\B$ is strong.
Define $i_{1}\colon Z_1 \emb X^0$ as the supremum of $U_{1}$ in $\Emb X^0$, and $i_{2}\colon Z_2 \emb Y^0$ as the supremum of $U_2$ in $\Emb Y^0$ (these suprema exist by Lemma~\ref{l:arboreal-facts}\ref{i:completely-join-irr}).
Then 
\[
f\coloneqq q_X\circ i_{1}\colon Z_{1}\to X \ \text{ and } \ g\coloneqq q_Y\circ i_{2}\colon Z_{2}\to Y
\] 
are pathwise embeddings because $q_X$ and $q_Y$ are pathwise embeddings.
Note that, since $U_{1}$ is downwards closed, by Lemma~\ref{l:arboreal-facts}\ref{i:completely-join-irr} any element of $\Path X^{0}$ that is below $i_{1}$ must belong to $U_{1}$; similarly for $i_{2}$ and $U_{2}$.

\vspace{0.5em}
\textbf{Third step: }
We claim that the pathwise embeddings $f$ and $g$ defined in the previous step are open.
We prove that $g$ is open; proving the corresponding fact for~$f$ is similar and slightly simpler. Consider a commuting diagram of the following form.

\[\begin{tikzcd}
           P \arrow[d,tail,"k" swap] \arrow[r,tail,"e"] & Z_{2} \arrow[d,"g"] \\
           Q \arrow[r,tail,"n'"] & Y
\end{tikzcd}\]
As $i_{2}\circ e\in \Path Y^{0}$ is below $i_{2}$, there exists $\brp{m,n} \in \B$ such that $i_{2}\circ e=\langle m,n\rangle_r$. Then
\begin{align*}
g\circ e &= q_Y\circ i_{2} \circ e \tag*{definition of $g$}\\
&=q_Y \circ \langle m,n\rangle_r \\
&= n \tag*{by~\eqref{eq:parallel-embeddings}}
\end{align*}
and so $n=n'\circ k$ implies that $n'\geq n$. Applying condition~\ref{i:pos-back-forth-back} in Definition \ref{d:pos-backforth}, possibly finitely many times, yields $m'\geq m$ such that $\brp{m',n'} \in \B$.
In particular, $\langle m',n'\rangle_r\colon Q\emb Y^0$ belongs to $U_{2}$ and so it is below $i_{2}$, i.e.\ there exists $h\colon Q\to Z_{2}$ such that $\langle m',n'\rangle_r = i_{2}\circ h$. We claim that $h$ is a diagonal filler for the square above, i.e.\ the following diagram commutes.
\[\begin{tikzcd}
P \arrow[d,tail,"k" swap] \arrow[r,tail,"e"] & Z_{2} \arrow[d,"g"] \\
Q \arrow[r,tail,"n'"] \arrow[ur,"h" description] & Y
\end{tikzcd}\]
The lower triangle commutes because $g\circ h = q_Y\circ i_{2} \circ h = q_Y\circ \langle m',n'\rangle_r = n'$. 
For the upper triangle, suppose for a moment that $\langle m,n\rangle_r = \langle m',n'\rangle_r\circ k$. Then 
\[
i_{2}\circ h\circ k = \langle m',n'\rangle_r \circ k = \langle m,n\rangle_r = i_{2}\circ e
\] 
and, since $i_{2}$ is monic, we get $h\circ k = e$ as desired.

It remains to show that $\langle m,n\rangle_r = \langle m',n'\rangle_r\circ k$. Since there is at most one embedding between any two paths, this amounts to showing that $\langle m,n\rangle_r \leq \langle m',n'\rangle_r$. In turn, since the map $\exists_{p_{Y}}\colon \Emb(X\times Y)\to\Emb{Y^{0}}$ is monotone, it is enough to prove that $\langle m,n\circ \lambda \rangle \leq \langle m',n'\circ \lambda'\rangle$, where $\lambda\colon \dom(m)\epi \dom(n)$ and $\lambda'\colon \dom(m')\epi \dom(n')$ are the unique quotients. Since $m\leq m'$, there is $l\colon \dom(m)\to \dom(m')$ such that $m = m'\circ l$. Using the fact that any diagram of arrows between paths commutes, it is easy to see that $\langle m',n'\circ \lambda'\rangle \circ l = \langle m,n\circ \lambda \rangle$, showing that $\langle m,n\circ \lambda \rangle \leq \langle m',n'\circ \lambda'\rangle$.

\vspace{0.5em}
\textbf{Fourth step:} In this last step, we construct $h\colon Z_{1}\to Z_{2}$ such that $\Path{h}$ is an isomorphism. This arrow is obtained as the unique mediating morphism induces by a compatible cocone with vertex $Z_{2}$ on the diagram of path embeddings into $Z_{1}$. Identifying $\Path{Z_{1}}$ and $\Path{Z_{2}}$ with sub-trees of $\Path{X^{0}}$ and $\Path{Y^{0}}$, respectively, the elements of $\Path{Z_{1}}$ are precisely those of the form $\langle m_i,n_i\rangle_l$ for $\brp{m_{i},n_{i}} \in \B$, and those of $\Path{Z_{2}}$ are precisely those of the form $\langle m_i,n_i\rangle_r$. For each $i$ there is a quotient
\[
h_{i}\colon \dom(\langle m_i,n_i\rangle_l) = \dom(m_{i}) \epi \dom(n_{i}) = \dom(\langle m_i,n_i\rangle_r).
\]
We claim that the following cocone is compatible, where we identify $k_{i}$ with $\langle m_i,n_i\rangle_l$, and $l_{i}$ with $\langle m_i,n_i\rangle_r$.
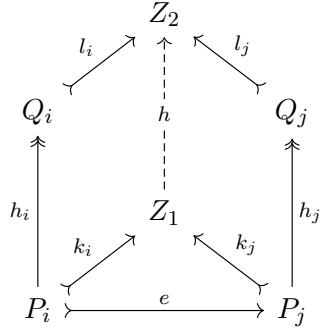
\[\begin{tikzcd}
       & Z_2 & \\
       Q_i \arrow[ru,tail,"l_i"]& & Q_j \arrow[lu,tail,"l_j" swap]\\
       & Z_1 \arrow[dashed]{uu}[description]{h} & \\
       P_i \arrow[uu,two heads, "h_i"] \arrow[ru,tail,"k_i"]\arrow[rr,tail,"e"]& & P_j \arrow[lu,tail,"k_j" swap] \arrow[uu,two heads, "h_j" swap]
\end{tikzcd}\]
The compatibility condition follows at once if we prove that the assignment $\langle m_i,n_i\rangle_l \mapsto \langle m_i,n_i\rangle_r$ is monotone. Recall that the diagram
\[\begin{tikzcd}
\Path{X^{0}} & \Path(X\times Y) \arrow{l}[swap]{\Path{p_{X}}} \arrow{r}{\Path{p_{Y}}} & \Path{Y^{0}}
\end{tikzcd}\]
consists of two isomorphisms of trees. The image of $U=\{\langle m, n\circ \lambda \rangle \mid \brp{m,n} \in \B\}$ under $\Path{p_{X}}$ and $\Path{p_{Y}}$ coincides with $\Path{Z_{1}}$ and $\Path{Z_{2}}$, respectively. Thus the tree isomorphisms above descend to isomorphisms 
\[\begin{tikzcd}
\Path{Z_{1}} & U \arrow{l}[swap]{\Path{p_{X}}} \arrow{r}{\Path{p_{Y}}} & \Path{Z_{2}}.
\end{tikzcd}\]
The assignment $\langle m_i,n_i\rangle_l \mapsto \langle m_i,n_i\rangle_r$ coincides with $\Path{p_{Y}}\circ (\Path{p_{X}})^{-1}$, so it is clearly monotone.
By construction, the induced morphism $h\colon Z_{1}\to Z_{2}$ satisfies $\Path{h} = \Path{p_{Y}}\circ (\Path{p_{X}})^{-1}$ and so it is an isomorphism of trees.
\end{proof}

\subsection{Positive arboreal games}\label{ss:pos-arbor-games}

We show that the positive arboreal games $\game^+(X,Y)$, when considered in the arboreal categories $\EM(\EkI)$, $\EM(\PkI)$ and $\EM(\Mk)$, capture precisely the positive $k$-round \EF game, the positive $k$-pebble game and the positive $k$-round bisimulation game, respectively. The proofs are similar to those in Section~\ref{s:ex-arbor}; we describe the main ingredients for Proposition~\ref{p:pos-arbor-EF} and leave the proofs of Propositions~\ref{p:pos-abor-peb} and~\ref{p:pos-arbor-mod} to the reader.

\begin{proposition}\label{p:pos-arbor-EF}
    Duplicator has a winning strategy in the positive arboreal game $\game^+(\FI_k(A),\FI_k(B))$ in the category $\EM(\EkI)$ if, and only if, they have a winning strategy in the positive $k$-round \EF game from $A$ to $B$.
\end{proposition}

\begin{proof}
Let $C$ denote either of the structures $A$ or $B$.
As in the proof of Proposition~\ref{p:ex-arbor-EF}, we need to show that the isomorphism between the forest obtained by removing the root of $\Path(\FI_k(C))$, and the underlying forest of $\FI_k(C)$, yields a bijection between (non-initial) winning positions of $\game^+(\FI_k(A),\FI_k(B))$ and winning positions in the positive $k$-round \EF game from $A$ to $B$. 

    \begin{itemize}
        \item Let $(m,n)$ be a valid position in $\game^+(\FI_k(A),\FI_k(B))$. If $(m,n)$ is a winning position, then $\dom(m)$ and $\dom(n)$ are paths of the same height, say ${x_1\cvr \cdots \cvr x_i}$ and ${y_1\cvr \cdots \cvr y_i}$, respectively, and there is a homomorphism $h\colon \dom(m)\to\dom(n)$ sending $x_{j}$ to $y_{j}$ for all $j\in\{1,\ldots,i\}$. The position in the \EF game corresponding to $(m,n)$ is $(s,t)$, where
\[
s\coloneqq [\epsilon_{JA}m(x_1),\dots,\epsilon_{JA}x(p_i)]
\]
and
\[
t\coloneqq [\epsilon_{JB}nh(x_1),\dots,\epsilon_{JB}nh(x_i)].
\]
The morphism $m$ is an embedding, hence it reflects relations, while the homomorphism $n\circ h$ preserves them. Thus the interpretations of the relation symbols in $\sigma\cup \{I\}$ are preserved from $A$ to $B$. That is, $(s,t)$ is a winning position in the positive $k$-round \EF game.
        \item Conversely, suppose that $(m,n)$ is a valid position in $\game^+(\FI_k(A),\FI_k(B))$ such that the corresponding position in the positive $k$-round \EF game is winning. Then the domains of $m$ and $n$ have the same height, say they are of the form ${x_1 \cvr \cdots \cvr x_i}$ and ${y_1 \cvr \cdots \cvr y_i}$ respectively. It is not difficult to see that the winning condition for the positive \EF game implies that the forest morphism $\dom(m)\to\dom(n)$ that sends $x_{j}$ to $y_{j}$, for all $j\in \{1,\ldots,i\}$, is a homomorphism of structures. Thus, $(m,n)$ is a winning position in $\game^+(\FI_k(A),\FI_k(B))$. \qedhere
    \end{itemize}
\end{proof}

\begin{proposition}\label{p:pos-abor-peb}
    Duplicator has a winning strategy for the positive arboreal game $\game^+(\FI_k(A),\FI_k(B))$ in the category $\EM(\PkI)$ if and only if they have a winning strategy for the positive $k$-pebble game from $A$ to $B$.
\end{proposition}

\begin{proposition}\label{p:pos-arbor-mod}
    Duplicator has a winning strategy for the positive arboreal game $\game^+(F_k (A,a), F_k (B,b))$ in the category $\EM(\Mk)$ if and only if they have a winning strategy for the positive $k$-round bisimulation game from $(A,a)$ to $(B,b)$.
\end{proposition}

The axiomatic result in Theorem~\ref{th:game-bisim}, combined with the previous three propositions, yields Propositions~\ref{p:pos-EF-struct}, \ref{p:pos-peb-struct} and~\ref{p:pos-mod-struct}, respectively.

\section{A preservation result for modal logic}\label{s:preservation}
Recall that \emph{Lyndon's positivity theorem}~\cite{Lyndon1959} states that a first-order sentence in a vocabulary $\sigma$ is logically equivalent to a positive sentence if, and only if, it is preserved under surjective homomorphisms between $\sigma$-structures.
This section aims to establish the following equi-depth positivity theorem for modal logic, and its relativisation to finite Kripke models, using positive bisimulations and the results in Section~\ref{s:concrete-pos}.

\begin{theorem}\label{t:Lyndon-modal}
The following statements are equivalent for any modal formula~$\phi$ of modal depth at most~$k$ in a (multi)modal vocabulary:
\begin{enumerate}[label=(\roman*)]
\item\label{i:phi-preserved} $\phi$ is preserved under surjective homomorphisms between pointed Kripke models (respectively, between finite pointed Kripke models).
\item\label{i:phi-positive} $\phi$ is logically equivalent (respectively, logically equivalent over finite pointed Kripke models) to a positive modal formula of modal depth at most $k$.
\end{enumerate}
\end{theorem}

\begin{remark}
A proof of Theorem~\ref{t:Lyndon-modal} can be obtained, at least for unimodal vocabularies, from Lyndon's interpolation theorem for modal logic; cf.\ e.g.\ \cite{Gabbay1972} or the forthcoming volume~\cite{craig}. The proof-theoretic approach has the advantage of deriving the relativisation to finite Kripke models from the finite model property of modal logic, and the preservation of modal depth from the analysis of the syntactic shape of the interpolants. The novelty of our proof, on the other hand, is that it is model-theoretic in spirit, yet it is resource-sensitive and works in the finite as well.
\end{remark}

We shall need the following standard characterisation of model classes of positive modal formulas.

\begin{lemma}\label{l:upward-closed-positively-axiomat}
Let $\sigma$ be a (multi)modal vocabulary and let $\D\hookrightarrow \CSstar$ be a full subcategory of pointed Kripke models. The following statements are equivalent:
\begin{enumerate}[label=(\roman*)]
\item  $\D=\Mod(\psi)$ for some $\psi\in \PML_k$.
\item $\D$ is upwards closed with respect to the preorder $\IMP_{\PML_{k}}$, that is
\[
(A,a)\IMP_{\PML_{k}} (B,b) \ \text{ and } \ (A,a)\in \D \ \Longrightarrow \ (B,b)\in \D.
\]
\end{enumerate}
\end{lemma}
\begin{proof}
This well-known fact follows from the existence of finitely many formulas in $\PML_k$ up to logical equivalence. It can be deduced from the corresponding statement for first-order formulas with bounded quantifier rank, see e.g.\ \cite[Lemma~3.13]{Libkin2004}.
\end{proof}

\begin{proof}[Proof of Theorem~\ref{t:Lyndon-modal}]
\ref{i:phi-preserved} $\Rightarrow$ \ref{i:phi-positive} If $\phi$ is preserved under surjective homomorphisms between pointed Kripke models, then the category $\Mod(\phi)$ is closed in $\CSstar$ under surjective images and saturated under the equivalence relation $\equiv_{\ML_{k}}$. In view of Lemma~\ref{l:upward-closed-positively-axiomat}, it suffices to show that $\Mod(\phi)$ is upwards closed with respect to $\IMP_{\PML_{k}}$. 

Suppose that $(A,a)\IMP_{\PML_{k}} (B,b)$. By Theorem~\ref{t:positive-concrete-charact}\ref{i:positive-modal-char}, there is a positive bisimulation $F_k(A,a) \pbisim F_k(B,b)$ in $\EM(\Mk)$, where the latter can be identified with the full subcategory of $\CSstar$ defined by synchronization trees of height $\leq k$. Such a positive bisimulation is of the form
\[\begin{tikzcd}[column sep=1.5em]
{} & Z \arrow{dl}[swap]{f} \arrow{dr}{g} & {} \\
F_{k}(A,a) & & F_{k}(B,b)
\end{tikzcd}\]
where $f$ is a p-morphism of Kripke models and $g$ is a homomorphism of Kripke models that is a p-morphism between the underlying tree orders; in particular $g$ is surjective. Therefore, if $\Mod(\phi)$ contains $(A,a)$ then it contains its $k$-unravelling $F_{k}(A,a)$ because $(A,a)\equiv_{\ML_{k}}F_{k}(A,a)$, and so it contains also $Z$ because p-morphisms (preserve and) reflect the validity of modal formulas. In turn, since $g$ is surjective we get $F_{k}(B,b)\in \Mod(\phi)$, and reasoning as before we conclude that $(B,b)\in \Mod(\phi)$.

The same proof, mutatis mutandis, yields the relativisation to finite Kripke models. Just observe that if $(A,a)$ and $(B,b)$ are finite, then so are $F_{k}(A,a)$ and $F_{k}(B,b)$. The Kripke model $Z$ is also finite, as it can be identified with the object $Z_{1}$ constructed in the fourth step of the proof of Proposition~\ref{p:backforth-bisim}. Concretely, the underlying set of $Z$ is a subset of the underlying set of $F_{k}(A,a)\times F_{k}(B,b)$, hence it is finite.

\ref{i:phi-positive} $\Rightarrow$ \ref{i:phi-preserved} This is a straightforward verification.
\end{proof}

\appendix

\section{Omitted proofs of logical equivalences via games}\label{app:logical-equiv}

In this appendix we provide detailed proofs for the results in Sections~\ref{s:concrete-exs} and~\ref{s:concrete-pos} whose proofs were either sketched or omitted.

\propexEFlogic*

\begin{proof}
\ref{i:exist-pres-EF} $\Rightarrow$ \ref{i:win-exis-EF}
Let $C\in \{A,B\}$. For $i\geq 0$ and $c_1,\dots,c_i\in C$, recall that the $\EFO_{k}$-type of $c_1,\dots,c_i$ is defined as
\[\tp_C(c_1,\dots,c_i)\coloneqq \{\phi(x_1,\dots,x_i) \in \EFO_{k-i} \mid C\models \phi (c_1,\dots,c_i)\}.
\] 
In particular $\tp_C(\bot)=\{\phi(\varnothing) \in \EFO_k \mid C\models \phi \}$, hence condition~\ref{i:exist-pres-EF} states that $\tp_A(\bot)\subseteq \tp_B(\bot)$.

We claim that, given $i\geq 1$, $a_1,\dots,a_{i-1}$ and $b_1,\dots,b_{i-1}$ such that $\tp_A(a_1,\dots,a_{i-1})\subseteq \tp_B(b_1,\dots,b_{i-1})$, if $a_i\in A$ there exists $b_i\in B$ such that $\tp_A(a_1,\dots,a_i)\subseteq \tp_B(b_1,\dots,b_i)$.
Let $\Phi\coloneqq \bigwedge{\tp_A(a_1,\dots,a_i)}$ and consider the formula 
\[
\psi(x_1,\dots,x_{i-1})\coloneqq \exists x_i \Phi(x_1,\dots,x_{i-1},x_i).
\]
Then $a_i$ witnesses the fact that $A\models \psi(a_1,\dots,a_{i-1})$.  Since $\psi$ is an element of $\EFO_{k-i+1}$ it belongs to $\tp_A(a_1,\dots,a_{i-1})$, and so also to $\tp_B(b_1,\dots,b_{i-1})$ by inductive hypothesis. That is, there exists $b_i\in B$ such that $B\models \Phi (b_1,\dots,b_i)$.
As each $\phi\in \tp_A(a_1,\dots,a_i)$ is equivalent to one of the conjuncts in $\Phi$, this yields $\tp_A(a_1,\dots,a_i)\subseteq \tp_B(b_1,\dots,b_i)$.
    
Now, let us consider the strategy in which, if elements $(a_{1},b_{1}), \ldots, (a_{i-1},b_{i-1})$ have been played and Spoiler chooses an element $a_{i}$, then Duplicator responds by choosing any $b_{i}$ that satisfies the property above. By induction, we have $\tp_A(a_1,\dots,a_i)\subseteq \tp_B(b_1,\dots,b_i)$. Let us show that this is a winning strategy for Duplicator.

The inclusion of types on formulas of the form $\neg (x_i=x_j)$ ensures that $\{(a_i,b_i)\}_{1\leq i\leq k}$ defines a functional and injective relation. Suppose now that $m(1),\dots,m(p)\leq k$.
\begin{itemize}
    \item If $R^A(a_{m(1)},\dots,a_{m(p)})$ then $R(x_{m(1)},\dots,x_{m(p)})\in \tp_A(a_1,\dots,a_j)$. Since the latter type is contained in $\tp_B(b_1,\dots,b_j)$, we get $R^B(b_{m(1)},\dots,b_{m(p)})$.
    
    \item Conversely, suppose that $R^A(a_{m(1)},\dots,a_{m(p)})$ does not hold. Then the formula $\neg R(x_{m(1)},\dots,x_{m(p)})$ belongs to $\tp_A(a_1,\dots,a_j)$, and by the same reasoning as above (which applies because this is indeed a formula of $\EFO_k$, since negations are allowed at atomic level) we conclude that $R^B(b_{m(1)},\dots,b_{m(p)})$ does not hold.
\end{itemize}
This shows that $\{(a_i,b_i)\}_{1\leq i\leq k}$ is a partial isomorphism, so Duplicator wins.

\ref{i:win-exis-EF} $\Rightarrow$ \ref{i:exist-pres-EF}
Let $\phi(\varnothing) \in \ELL_k$. Since the quantifier rank is at most $k$, we can rename the (bound) variables of $\phi$ to $x_1,\dots,x_k$, making sure that the quantifiers $\exists x_i$ appear in $i$-increasing order within a quantifier branch. 
Let us prove by induction on $\psi$ that, given $0\leq i\leq k$, $a_1,\dots,a_i$ and $\psi(x_1,\dots,x_i)$ subformula of $\phi$, if $A\models \psi (a_1,\dots,a_i)$ then $B\models \psi (b_1,\dots,b_i)$, where $[b_1,\dots,b_i]$ is Duplicator's answer to Spoiler's play $[a_1,\dots,a_n]$.

\begin{enumerate}
    \item If $\psi$ is a disjunction or conjunction of formulas that satisfy the inductive hypothesis, the desired statement follows at once.
    
    \item If $\psi$ is of the form $R(x_{m(1)},\dots,x_{m(p)})$ then, since $A\models \psi (a_1,\dots,a_i)$, we have $R^A(a_{m(1)},\dots,a_{m(p)})$ and thus $R^B(b_{m(1)},\dots,b_{m(p)})$. That is, $B\models \psi (b_1,\dots,b_i)$. 
    
    \item If $\psi$ is $x_{m}=x_{m'}$, using the fact that the relation is functional, if $a_m=a_{m'}$ then $b_{m}=b_{m'}$.
    
    \item If $\psi$ is of the form $\neg R(x_{m(1)},\dots,x_{m(p)})$ or $\neg (x_m = x_{m'})$, then the proof is the same, mutatis mutandis, and uses the fact that relations are reflected and the functional relation is injective, respectively.
    
    \item Finally, if $\psi$ is $\exists x_{i+1}\chi$ and $A\models \psi (a_1,\dots,a_i)$, there exists $a_{i+1}\in A$ such that $A\models \chi(a_1,\dots,a_{i+1})$. Let $b_{i+1}$ be Duplicator's answer to the play $[a_1,\dots,a_{i+1}]$. As $\chi$ is a subformula of $\phi$ where $x_1,\dots,x_{i+1}$ are free, by inductive hypothesis we have $B\models \chi(b_1,\dots,b_{i+1})$ and thus $B\models \psi (b_1,\dots,b_i)$.

\end{enumerate}

Now, if $A\models \phi$, with the induction property applied to $\phi$ and $i=0$, we get $B\models \phi$.
\end{proof}

\propexpeblogic*

\begin{proof}
Given $[(p_1,c_1),\dots,(p_n,c_n)]\in \FI_k(A)$, fix an enumeration $r_1,\dots,r_m$ of the $p_i$'s (\ie $\{r_1,\dots,r_m\}=\{p_1,\dots,p_n\}$ and the $r_j$ are pairwise distinct) and a function $l$ such that $l(r_j)\leq n$ is the index of the last appearance of pebble $r_j$ in the sequence $[p_1,\dots,p_n]$.

\ref{i:exist-pres-pebble} $\Rightarrow$ \ref{i:win-exis-peb} 
For $p_1,\dots, p_n\leq k$ and $c_1,\dots,c_n\in C$, the $\exists \LL^k$-type of $(p_1,c_1),\dots,(p_n,c_n)$ in $C$ is 
\[
\tp_C((p_1,c_1),\dots,(p_n,c_n))\coloneqq \{\psi(x_{r_1} ,\dots,x_{r_m} )\in \ELL^k\mid C\models \psi  (c_{l(r_1)},\dots,c_{l(r_m)})\}.
\] 
We have $\tp_C(\bot)=\{\phi(\varnothing) \in \ELL^k \mid C\models \phi\}$ thus item~\ref{i:exist-pres-pebble} states that $\tp_A(\bot)\subseteq \tp_B(\bot)$.

Let us prove the following property: Given $i\geq 1$, $p_1,\dots, p_{i-1}, a_1,\dots,a_{i-1}$ and $b_1,\dots,b_{i-1}$ such that $\tp_A((p_1,a_1),\dots,(p_{i-1}, a_{i-1}))\subseteq \tp_B(b_1,\dots,b_{i-1})$, for all ${1\leq p_i \leq k}$ and $a_i\in A$, there is $b_i\in B$ such that $\tp_A((p_1,a_1),\dots,(p_i,a_i))\subseteq \tp_B((p_1,b_1),\dots,(p_i,b_i))$.

 Suppose for simplicity that in the enumeration $r_1,\dots,r_m$ of the $p_j$ we have $r_m=p_i$ (and thus $l(r_m)=i$).
    Let $\Phi\coloneqq \bigwedge{\tp_A((p_1,a_1),\dots,(p_i,a_i))}$ and note that $\Phi$ belongs to $\ELL^k$ since we allow infinitary conjunctions. If 
\[
\psi(x_{r_1},\dots, x_{r_{m-1}})\coloneqq \exists x_{p_i} \Phi(x_{l(r_1)},\dots,x_{l(r_{m-1})},x_{p_i}).
\]
then $a_i$ witnesses the fact that $A\models \psi(a_{l(r_1)},\dots,a_{l(r_{m-1})})$.

    \begin{itemize}
        \item If $p_i$ does not appear in $p_1,\dots,p_{i-1}$, then the formula $\psi(x_{r_1},\dots, x_{r_{m-1}})$ belongs to $\tp_A((p_1,a_1),\dots,(p_{i-1} ,a_{i-1}))$. By inductive hypothesis, it also belongs to $\tp_B((p_1,a_1),\dots,(p_{i-1} ,a_{i-1}))$, which entails the existence of $b_i\in B$ such that $B\models \Phi (b_{l(r_1)},\dots,b_{l(r_{m-1})},b_i)$. 
Since each $\phi\in \tp_A((p_1,a_1),\dots,(p_i,a_i))$ is equivalent to one of the conjuncts in $\Phi$, this yields $\tp_A((p_1,a_1),\dots,(p_i,a_i))\subseteq \tp_B((p_1,b_1),\dots,(p_i,b_i))$.

        \item Now, suppose $p_i=r_m$ appears in $p_1,\dots,p_{i-1}$, and let its last appearance be at index $l'(p_i)$. We can view $\psi$ as a formula over the variables $x_{r_1},\dots,x_{r_m}$ which does not depend on $x_{r_m}$, hence the previous argument applies.
        We thus have $A\models \psi(a_{l(r_1)},\dots,a_{l'(r_m)})$. This means that $\psi\in  \tp_A((p_1,a_1),\dots,(p_{i-1} ,a_{i-1}))$, and we can conclude by reasoning as above. 
    \end{itemize}

We claim that the strategy consisting in responding with the so-defined $b_i$ to the play $[(p_1,a_1),\dots,(p_i,a_i)]$ is winning for Duplicator. Suppose that $[(p_1,a_1),\dots,(p_n,a_n)]$ is Spoiler's play, and let Duplicator's answers be $[(p_1,b_1),\dots,(p_n,b_n)]$. As proved above, we have, by induction, that $\tp_A((p_1,a_1),\dots,(p_n,a_n))\subseteq \tp_B((p_1,b_1),\dots,(p_n,b_n))$.
Consider the last positions of the pebbles $a_{l(r_1)},\dots,a_{l(r_m)}$ and $b_{l(r_1)},\dots,b_{l(r_m)}$. To see that these form the graph of an injective function, suppose $a_{l(r_i)}= a_{l(r_j)}$. The formula $x_{r_i}=x_{r_j}$ is an element of $\tp_A((p_1,a_1),\dots,(p_n,a_n))$ thus it also belongs to $\tp_B((p_1,b_1),\dots,(p_n,b_n))$, which means that $b_{l(r_i)}= b_{l(r_j)}$. This shows that the relation is functional. The same argument, applied to the formula $\neg (x_{r_i}=x_{r_j})$, shows that the function is injective. In order to show that this is a partial isomorphism, let $m(1),\dots,m(q)\leq m$.
\begin{itemize}
    \item If $R^A(a_{l(r_{m(1)} )},\dots,a_{l(r_{m(q)})})$ then $R(x_{r_{m(1)} },\dots,x_{r_{m(q)}})\in \tp_A(a_1,\dots,a_j)$. Thus $R(x_{r_{m(1)} },\dots,x_{r_{m(q)}})\in \tp_B(b_1,\dots,b_j)$ and so $R^B(b_{l(r_{m(1)} )},\dots,b_{l(r_{m(q)})})$.
    
    \item For the converse, if $R^A(a_{l(r_{m(1)} )},\dots,a_{l(r_{m(q)})})$ does not hold then the formula $\neg R(x_{r_{m(1)} },\dots,x_{r_{m(q)}})$ belongs to $\tp_A(a_1,\dots,a_j)$ and the same reasoning as above applies (recall that in formulas of $\ELL^k$ negations are allowed at the atomic level).
\end{itemize}

This proves that $\{(a_{l(r_i)} ,b_{l(r_i)})\}_{1\leq i\leq m}$ is a partial isomorphism, so Duplicator wins the existential pebble game.

\ref{i:win-exis-peb} $\Rightarrow$ \ref{i:exist-pres-pebble}
Suppose that $\phi(\varnothing) \in \ELL^k$ satisfies $A\models \phi$. Since $\phi$ has only $k$ variables, we can rename them to $x_1,\dots,x_k$.
We prove by induction on subformulas $\psi(x_{r_1},\dots,x_{r_m})$ of $\phi$ that the following property holds:
Given a positive integer $n$ and $[(p_1,a_1),\dots,(p_n,a_n)]\in \FI_k(A)$, if $b_i$ is Duplicator's response to $[(p_1,a_1),\dots,(p_i,a_i)]$ and $A\models \psi (a_{l(r_1)},\dots,a_{l(r_m)})$, then $B\models \psi (b_{l(r_1)},\dots,b_{l(r_m)})$.

\begin{enumerate}
    \item If $\psi$ is a (possibly infinite) disjunction or conjunction of formulas that satisfy the inductive hypothesis, the desired statement follows at once.
  
    \item If $\psi = R(x_{r_{\sigma_1}},\dots,x_{r_{\sigma_q}})$ with $\sigma_1,\dots,\sigma_q \leq m$ and $A\models \psi (a_{l(r_1)},\dots,a_{l(r_m)})$, then $R^A(a_{l(r_{\sigma_1})},\dots,a_{l(r_{\sigma_q})})$. Given that Duplicator's strategy is winning, we get $R^B(b_{l(r_{\sigma_1})},\dots,b_{l(r_{\sigma_q})})$, \ie $B\models \psi (b_{l(1)},\dots,b_{l(k)})$.    
    
    \item If $\psi$ is $x_{r_i}=x_{r_j}$ and $A\models \psi (a_{l(r_1)},\dots,a_{l(r_m)})$ then $a_{l(r_i)}=a_{l(r_j)}$. The relation $\{(a_{l(r_i)} ,b_{l(r_i)})\}_{1\leq i\leq m}$ is functional, so $b_{l(r_i)}=b_{l(r_j)}$. \Iec $B\models \psi (b_{l(r_1)},\dots,b_{l(r_m)})$. 

    \item If $\psi$ is $\neg R(x_{r_{\sigma_1}},\dots,x_{r_{\sigma_q}})$ or $\neg (x_{r_i}=x_{r_j})$, the proof is the same, mutatis mutandis, and hinges on the fact that relations are reflected and the functional relation $\{(a_{l(r_i)} ,b_{l(r_i)})\}_{1\leq i\leq m}$ is injective.
    
    \item Suppose $\psi$ is $\exists x_i \chi$ and let $A\models \psi (a_{l(r_1)},\dots,a_{l(r_m)})$. We distinguish two cases:
    
    \begin{itemize}
        \item If $i$ does not appear in $r_1,\dots,r_m$ then $\chi = \chi(x_{r_1},\dots, x_{r_m},x_i)$. Set $r_{m+1}\coloneqq i$ and $l(r_{m+1} )\coloneqq n+1$, and let $a_{n+1}$ satisfy $A\models \chi(a_{l(r_1)},\dots,a_{l(r_m)},a_{l(r_{m+1})})$.

        Let $b_1,\dots,b_n,b_{n+1}$ be Duplicator's answers to Spoiler's play 
        \[
        [(p_1,a_1),\dots,(p_n,a_n),(i,a_{n+1})].
        \] 
        By inductive hypothesis applied to $\chi$ and the sequences $p_1\dots,p_n,i$ and $a_1,\dots,a_n,a_{n+1}$, we have $B\models \chi(b_{l(r_1)},\dots,b_{l(r_m)},b_{l(r_{m+1})})$.
        Thus, $b_{n+1}$ is a witness that $B\models \psi (b_{l(r_1)},\dots,b_{l(r_m)})$.
        
        \item If $i=r_j$ for some $j\in\{1,\ldots,m\}$, we set $l'(r_j)\coloneqq n+1$ and $l'(r_q)\coloneqq l(r_q)$ for $q\neq j$. Then $\chi=\chi(x_{r_1},\dots,x_{r_m})$ and since the variable $x_{r_j}$ is in the scope of the outer quantifier $\exists x_i$ in $\psi$, the substitution of the variable $x_{r_j}$ in $\psi$ has no effect. 
Again there exists $a_{n+1}$ such that $A\models \chi(a_{l'(r_1}),\dots,a_{l'(r_m)})$, and by inductive hypothesis this implies that $B\models \chi(b_{l'(r_1}),\dots,b_{l'(r_m)})$ where the $b_i$'s are Duplicator's answers. 
Thus $b_{n+1}$ is a witness that $B\models \psi (b_{l(r_1)},\dots,b_{l(r_m)})$.
    \end{itemize}
\end{enumerate}
Applying the property above to $\phi$, with $n=0$, we get $B\models \phi$ as desired.
\end{proof}

\propexmodlogic*

\begin{proof}
\ref{i:EMLk-preserved} $\Rightarrow$ \ref{i:Dup-wins-k-ex-bisim}
Define $\tp^i_A(a) \coloneqq \{\phi \in \EML_{k-i}\mid A \models \phi (a)\}$; $\tp^i_B(b)$ is defined similarly. Then $\tp^0_A(a)\subseteq \tp^0_B(b)$ by~\ref{i:EMLk-preserved}. 
We claim that if $\tp^{i-1}_A(a_{i-1} )\subseteq \tp^{i-1}_B(b_{i-1})$, for any $R_i$ and $a_{i}$ such that $R_i^{A}(a_{i-1},a_i)$ there exists $b_i\in B$ such that $R_i^{B}(b_{i-1},b_i)$ and $\tp^i_A(a_{i})\subseteq \tp^i_B(b_{i})$.
This is proved by setting 
\[
\psi(x_{i-1}) \coloneqq \exists x_i \Big(R_i(x_{i-1},x_i)\wedge \Big(\bigwedge_{\phi \in \tp^i_A(a_{i})} \phi(x_i) \Big)\Big),
\] 
which is the standard translation of $\Diamond_{R_i}\bigwedge{\tp^i_A(a_{i})}$, if $\tp^i_A(a_{i})$ is regarded as a collection of modal formulas. 
With this notation we have $\psi \in \tp^{i-1}_A(a_{i-1}) \subseteq \tp^{i-1}_B(b_{i-1})$, and so the claim follows as in the previous propositions.

Answering the so-defined element $b_i$ yields a winning strategy for Duplicator: we have $R_i^B(b_{i-1},b_i)$ by construction, and since $\tp^i_A(a_{i})\subseteq \tp^i_B(b_{i})$ for all $1\leq i \leq k $, if $P^A(a_i)$ then $P(x)\in \tp^i_A(a_{i})$ and so $P^B(b_i)$. Similarly, $\neg P^A(a_i)$ implies $\neg P^B(b_i)$.

\ref{i:Dup-wins-k-ex-bisim} $\Rightarrow$ \ref{i:EMLk-preserved}
 Let $\phi (x_0) \in \EML_k$. We prove by induction on subformulas $\psi (x_i)$ of $\phi (x_0)$ that for any $a=a_0,\dots, a_i$, if $b=b_0,\dots, b_i$ is Duplicator's answer to this play, $A\models \psi (a_i)$ entails $B\models \psi (b_i)$.
 Applying this property to $\phi$ with $i=0$ will yield the desired result.

  If $\psi$ is a disjunction or conjunction of formulas that satisfy the inductive hypothesis, the desired statement follows easily. If $\psi$ is $P(x_i)$ or $\neg P(x_i)$, the claim follows at once from the fact that Duplicator's strategy is winning.
    If $\psi$ is $\exists x_{i+1} R(x_i,x_{i+1}) \wedge \chi (x_{i+1})$ and $A\models \psi (a_i)$, there exists $a_{i+1}$ such that $R^A(a_i,a_{i+1})$ and $A\models \chi(a_{i+1})$. Let $b_{i+1}$ be Duplicator's answer to the play $a_0,\dots,a_i,a_{i+1}$, and let $b_0,\dots,b_i$ be their previous answers. Then $R^B(b_i,b_{i+1})$ and by inductive hypothesis we have $B\models \chi(b_{i+1})$, which concludes the proof.
\end{proof}

\propposEFlogic*

\begin{proof}
\ref{i:PFOk-preserv} $\Rightarrow$ \ref{i:Dupl-wins-positive-k-game}
For $C\in \{A,B\}, 0\leq i \leq k$ and $c_1,\dots,c_i\in C$, define 
\[
\tp_C(c_1,\dots,c_i)\coloneqq \{\phi(x_1,\dots,x_i)\in \PFO_{k-i} \mid C\models \phi(c_1,\dots,c_i)\}.
\] 
Then $\tp_C(\bot)=\{\phi(\varnothing) \in \PFO_k \mid C\models \phi\}$ and so item~\ref{i:PFOk-preserv} implies $\tp_A(\bot)\subseteq \tp_B(\bot)$.

Let us prove by induction on $1\leq i \leq k$ that given $a_1,\dots,a_{i-1}\in A$ and ${b_1,\dots,b_{i-1}\in B}$ such that $\tp_A(a_1,\dots,a_{i-1})\subseteq \tp_B(b_1,\dots,b_{i-1})$,

\begin{enumerate}  
    \item\label{i:forth-1} For all $a_i\in A$ there exists $b_i\in B$ such that $\tp_A(a_1,\dots,a_{i})\subseteq \tp_B(b_1,\dots,b_{i})$.
    
    \item\label{i:back-2} For all $b_i\in B$ there exists $a_i\in A$ such that $\tp_A(a_1,\dots,a_{i})\subseteq \tp_B(b_1,\dots,b_{i})$.
\end{enumerate}
For item~\ref{i:forth-1}, let $a_i\in A$ and define $\Phi(x_1,\dots,x_{i-1}) \coloneqq \exists x_i \bigwedge{\tp_A(a_1,\dots,a_i)}$. Note that $\Phi\in \PFO_{k-i+1}$ because $\tp_A(a_1,\dots,a_i)\subseteq \PFO_{k-i}$.
Moreover, $\Phi \in \tp_A(a_1,\dots,a_{i-1})\subseteq \tp_B(b_1,\dots,b_{i-1})$. 
As in Proposition~\ref{p:ex-EF-logic}, this shows that $B\models \Phi(b_1,\dots,b_{i-1})$ and thus there exists $b_i\in B$ such that $B\models \phi (b_1,\dots b_i)$ for every $\phi \in \tp_A(a_1,\dots,a_i)$. Hence $\tp_A(a_1,\dots,a_{i})\subseteq \tp_B(b_1,\dots,b_{i})$.

For item~\ref{i:back-2}, let $b_i\in B$ and consider $S\coloneqq \PFO_{k-i} \setminus \tp_B(b_1,\dots,b_i)$.
If 
\[
\psi(x_1,\dots,x_i)\coloneqq \bigvee{S} \ \text{ and } \ \Phi(x_1,\dots,x_{i-1})\coloneqq \forall x_i \psi \in \PFO_{k-i+1},
\] 
then $B\not\models \psi(b_1,\dots,b_i)$ and so $B\not\models \Phi(b_1,\dots,b_{i-1})$. But $\Phi \notin \tp_B(b_1,\dots,b_{i-1})$ implies $\Phi \notin \tp_A(a_1,\dots,a_{i-1})$ by inductive hypothesis and so $A\not\models \Phi (a_1,\dots,a_{i-1})$.
This yields $a_i$ such that $A\not\models \phi(a_1,\dots,a_i)$ for each $\phi\in S$. Thus $\tp_A(a_1,\dots,a_{i})\subseteq \tp_B(b_1,\dots,b_{i})$.

Now, Duplicator's strategy consisting in responding with the so-built $a_i$ or $b_i$, depending on whether Spoiler plays in $B$ or $A$, is winning. The proof is similar to the case of existential fragments: the property $\tp_A(a_1,\dots,a_{i})\subseteq \tp_B(b_1,\dots,b_{i})$ is preserved at each step. 
The latter inclusion applied to formulas of the form $x_i=x_j$ ensures that the relation $\{(a_i,b_i)\}_{1\leq i \leq k}$ is functional; the corresponding function preserves relations because of the inclusion of types applied to formulas $R(x_{i_1},\dots,x_{i_q})$.

\ref{i:Dupl-wins-positive-k-game} $\Rightarrow$ \ref{i:PFOk-preserv}
Let $\phi (\varnothing)\in \PFO_k$.
We prove the following property by induction on subformulas $\psi(x_1,\dots,x_i)$ of $\phi$:
If the sequences $[a_1,\dots,a_i]$ and $[b_1,\dots,b_i]$ correspond to a play in the positive \EF game in which Duplicator uses their winning strategy, then $A\models \psi(a_1,\dots,a_i)$ implies $B\models \psi(b_1,\dots,b_i)$.

\begin{itemize}
    \item The property holds if $\psi$ is atomic because the strategy is winning for Duplicator, and it is easily seen to hold if $\psi$ is a conjunction or disjunction of formulas that satisfy the property.
        
    \item If $\psi = \exists x_{i+1}\chi(x_1,\dots,x_{i+1})$, let $a_{i+1}\in A$ be a witness for $A\models \psi (a_1,\dots,a_i)$, and let $b_{i+1}\in B$ be Duplicator's answer in round $i+1$ after the sequences $[a_1,\dots,a_i]$, $[b_1,\dots,b_i]$ have been played. The sequences $[a_1,\dots,a_{i+1}]$ and $[b_1,\dots,b_{i+1}]$ correspond to a play according to Duplicator's winning strategy, so the inductive hypothesis applied to $\chi$ yields $B\models \chi(b_1,\dots,b_{i+1})$ and so $B\models \psi(b_1,\dots,b_i)$.
    
    \item If $\psi = \forall x_{i+1}\chi(x_1,\dots,x_{i+1})$ and $A\models \psi (a_1,\dots,a_i)$, the proof is similar: for any $b_{i+1}\in B$, let $a_{i+1}\in A$ be Duplicator's answer to this play in round $i+1$. Then $A\models \chi(a_1,\dots,a_{i+1})$ and so, by inductive hypothesis on $\chi$, $B\models \chi(b_1,\dots,b_{i+1})$. It follows that $B\models \psi (b_1,\dots,b_i)$.
\end{itemize}
The established property, applied to the sentence $\phi$, yields the result.
\end{proof}

\proppospeblogic*

\begin{proof}
    Again, given $p_1,\dots,p_n \leq k$, we fix an enumeration $r_1,\dots,r_m$ of the $p_i$'s and denote by $l(r_j)$ the index of the last appearance of pebble $r_j$ in $p_1,\dots,p_n$.

\ref{i:PLLK-preserved} $\Rightarrow$ \ref{i:Duplicator-wins-positive-k-peb}
We define the positive type 
\[
\tp_C((p_1,c_1),\dots,(p_n,c_n)) \coloneqq \{\phi(x_{r_1},\dots,x_{r_m}) \in \PLLk \mid C\models \phi(a_{l(r_1)},\dots,a_{l(r_m)})\}.
\]
By assumption, $\tp_A(\bot)\subseteq \tp_B(\bot)$. We show that for all $p_1,\dots,p_{n-1}\leq k$, $a_1,\dots,a_{n-1}\in A$, and $b_1,\dots,b_{n-1}\in B$ such that 
\[
\tp_A((p_1,a_1),\dots,(p_{n-1},a_{n-1}))\subseteq \tp_B((p_1,b_1),\dots,(p_{n-1},b_{n-1})),
\] 
the following two properties hold:
\begin{itemize}
\item For all $p_n\leq k$ and $a_n\in A$, there exists $b_n \in B$ such that 
\[
\tp_A((p_1,a_1),\dots,(p_{n},a_{n}))\subseteq \tp_B((p_1,b_1),\dots,(p_{n},b_{n})).
\]

\item For all $p_n\leq k$ and $b_n\in B$, there exists $a_n \in A$ such that 
\[
\tp_A((p_1,a_1),\dots,(p_{n},a_{n}))\subseteq \tp_B((p_1,b_1),\dots,(p_{n},b_{n})).
\]
    \end{itemize}
For the first property, one can define $\psi \coloneqq \exists x_{p_n} \bigwedge{\tp_A((p_1,a_1),\dots,(p_{n},a_{n}))}$, where the conjunction is infinitary, and reason as in existential part.
For the second one, let $S \coloneqq \PLLk \setminus \tp_B((p_1,b_1),\dots,(p_{n},b_{n}))$ and $\psi \coloneqq \forall x_{p_n} \bigvee{S}$, where the disjunction is infinitary. The reasoning in the proof of Proposition~\ref{p:pos-EF-logic} can then be adapted in a straightforward manner to yield the desired result.

    Now we have $\tp_A((p_1,a_1),\dots,(p_{n},a_{n})\subseteq \tp_B((p_1,b_1),\dots,(p_{n},b_{n})$ after any number $n$ of rounds, and the preservation of formulas of the form $x_{r_i}=x_{r_j}$ ensures that the last positions of the pebbles form a functional relation, and preservation of formulas $R(x_{r_{\sigma_1}},\dots,x_{r_{\sigma_q}})$ that it is a partial homomorphism. Thus Duplicator wins the game.

\ref{i:Duplicator-wins-positive-k-peb} $\Rightarrow$ \ref{i:PLLK-preserved}
Let $\phi(\varnothing) \in \PLLk$.
We show by induction on subformulas $\psi(x_{r_1},\dots,x_{r_m})$ of $\phi$ the following property: for any $n\geq 0$, $p_1,\dots,p_n\leq k$, $a_1,\dots,a_n \in A$ and $b_1,\dots b_n \in B$ such that 
\[
[(p_1,a_1),\dots,(p_n,a_n)], [(p_1,b_1),\dots,(p_n,b_n)]
\] 
is a valid play in the positive $k$-pebble game where Duplicator uses their winning strategy, if $A\models \psi(a_{l(r_1)},\dots,a_{l(r_m)})$ then $B\models \psi(b_{l(r_1)},\dots,b_{l(r_m)})$.

    This is true for atomic formulas because the strategy is winning for Duplicator, and for conjunctions or disjunctions of formulas satisfying the property above this is easily seen to hold. 
    If $\psi$ is $\exists x_i \chi$, there is $a_{n+1}\in A$ witnessing that $A\models \psi(a_{l(r_1)},\dots,a_{l(r_m)})$. Set $p_{n+1}=i$ and let $b_{n+1}$ Duplicator's answer to the play $(p_{n+1},a_{n+1})$ at round $n+1$. By inductive hypothesis, $b_{n+1}$ is a witness that $B\models \psi(b_{l(r_1)},\dots,b_{l(r_m)})$.
	Finally, if $\psi$ is $\forall x_i \chi$, consider an arbitrary $b_{n+1}\in B$ and let $a_{n+1}$ be Duplicator's answer to Spoiler's move $(i,b_{n+1})$. By hypothesis we have $A\models \chi(a_{l'(r'_1)},\dots,a_{l'(r'_{m'})})$ (see the proof of Proposition~\ref{p:ex-EF-logic} for the definitions of $l',r',m'$), which yields $B\models \chi(b_{l'(r'_1)},\dots,b_{l'(r'_{m'})})$. Thus $B\models \psi(b_{l(r_1)},\dots,b_{l(r_m)})$.

    Applying the property above to $\phi$, with $n=0$, concludes the proof.
\end{proof}

\propposmodlogic*

\begin{proof}
\ref{i:PMLk-preserved} $\Rightarrow$ \ref{i:Dupl-wins-k-bisim-game}
Define the positive type $\tp^i_A(a) \coloneqq \{\phi \in \PML_{k-i}\mid A \models \phi (a)\}$; $\tp^i_B(b)$ is defined in a similar way.
By assumption we have $\tp^0_A(a)\subseteq \tp^0_B(b)$. We claim that for all $a=a_0,\dots,a_{i-1}\in A$ and $b=b_0,\dots,b_{i-1}\in B$ such that $\tp^{i-1}_A(a_{i-1})\subseteq \tp^{i-1}_B(b_{i-1})$, the following hold:
    \begin{itemize}
        \item For every $R\in \sigma$ and $a_{i}\in A$ such that $R^A(a_{i-1},a_i)$, there exists $b_i\in B$ such that $R^B(b_{i-1},b_i)$ and $\tp^i_A(a_{i})\subseteq \tp^i_B(b_{i})$.

        \item For every $R\in \sigma$ and $b_{i}\in B$ such that $R^B(b_{i-1},b_i)$, there exists $a_i\in B$ such that $R^A(a_{i-1},a_i)$ and $\tp^i_A(a_0,\dots,a_{i})\subseteq \tp^i_B(b_0,\dots,b_{i})$.
    \end{itemize}
The proof is essentially the same as for Propositions~\ref{p:pos-EF-logic} and~\ref{p:pos-peb-logic}: apply the inductive hypothesis to the standard translations of $\Diamond_{R}\bigwedge \left( \tp^i_A(a_i) \right)$ and $\Box_R \bigvee \left( \PML_{k-i} \setminus\tp^i_B(b_i) \right)$, respectively. 
It is not difficult to see that the previous two properties allow us to define a Duplicator winning strategy.

\ref{i:Dupl-wins-k-bisim-game} $\Rightarrow$ \ref{i:PMLk-preserved}
    Let $\phi \in \PML_k$. We prove that for any subformula $\psi$ of $\phi$, and any $i\geq 0$, if $a=a_0,\dots,a_i$, $b=b_0,\dots,b_i$ corresponds to a valid play in the $k$-round bisimulation game where Duplicator uses their winning strategy, then $A\models \psi(a_i)$ implies $B\models \psi(b_i)$.

The atomic case holds because the strategy is winning for Duplicator, and the case where $\psi$ is a conjunction or disjunction of formulas satisfying this property is easily verified.
If $\psi$ starts with $\exists$ (corresponding to a modality $\Diamond_R$) and Spoiler plays a witness in $A$, then by inductive hypothesis Duplicator's answer in $B$ will be a witness for $\psi$.
If $\psi$ starts with $\forall$ (corresponding to a modality $\Box_R$) and Spoiler chooses an arbitrary ($R$-accessible) element of $B$, Duplicator answer in $A$ will make $\psi$ true. By the inductive hypothesis, $\psi$ will hold in $B$ too.

When applied to the trivial subformula $\phi$, this concludes the proof.
\end{proof}

\addtocontents{toc}{\protect\setcounter{tocdepth}{0}}
\section*{Acknowledgements}
The last-named author is grateful to Balder ten Cate for useful information on Lyndon's positivity theorem and interpolation for modal logic.

\addtocontents{toc}{\protect\setcounter{tocdepth}{1}}

%%%%%%%%%%   BIBLIOGRAPHY   %%%%%%%%%%

\bibliographystyle{amsplain-nodash}

\providecommand{\bysame}{\leavevmode\hbox to3em{\hrulefill}\thinspace}
\providecommand{\MR}{\relax\ifhmode\unskip\space\fi MR }
% \MRhref is called by the amsart/book/proc definition of \MR.
\providecommand{\MRhref}[2]{%
  \href{http://www.ams.org/mathscinet-getitem?mr=#1}{#2}
}
\providecommand{\href}[2]{#2}

\end{document}